\documentclass[a4paper]{amsart}
\usepackage{amsowariMT}
\pdfoutput=1
\usepackage{stmaryrd}
\usepackage{amsaddr_scauthor}
\usepackage[numbers]{natbib}
\usepackage{amsmath,mathrsfs,amssymb}		

\usepackage[OT1]{fontenc}
\usepackage[utf8]{inputenc}

\begin{document}

\title[Duality in Robust Utility Maximization]{Duality in Robust
  Utility Maximization with Unbounded Claim via a Robust Extension of
  Rockafellar's Theorem }

\author[K. Owari]{ Keita Owari}
\address{Institute of Economic Research,  Hitotsubashi University\\
  2-1 Naka, Kunitachi, Tokyo 186-8603, Japan }
\curraddr{}
\email{\href{mailto:keita.owari@gmail.com}{keita.owari@gmail.com}}



\dedicatory{}

\begin{abstract}
  We study the convex duality method for robust utility maximization
  in the presence of a random endowment.  When the underlying price
  process is a locally bounded semimartingale, we show that the
  fundamental duality relation holds true for a wide class of utility
  functions on the whole real line and unbounded random endowment. To
  obtain this duality, we prove a robust version of Rockafellar’s
  theorem on convex integral functionals and apply Fenchel’s general
  duality theorem.
\end{abstract}
\maketitle

\section{Introduction}
\label{intro}

We study the convex duality method for robust utility maximization
with a random endowment. Suppose we are given a semimartingale $S$
describing the evolution of the underlying asset prices, a class
$\Theta$ of admissible integrands (strategies) for $S$, a utility
function $U$, a set $\mathcal{P}$ of probability measures, and a random
variable (endowment) $B$ which is interpreted as the terminal payoff
of a contingent claim. Then the problem is to
\begin{equation}
  \label{eq:ProbUtil1}
  \text{maximize}\quad \inf_{P\in\mathcal{P}}E_P[U(\theta\cdot S_T+B)]\quad\text{over all }\theta\in\Theta.
\end{equation}
The set $\mathcal{P}$ is a mathematical formulation of \emph{model
  uncertainty} (also called \emph{Knightian uncertainty} in
economics), i.e., each element $P\in\mathcal{P}$ is considered as a
\emph{candidate model} of the financial market. We refer to
\cite{follmer_schied_weber09} for background information and recent
results on robust utility maximization problem.

A way of solving (\ref{eq:ProbUtil1}) is the convex duality method
which pass (\ref{eq:ProbUtil1}) to a minimization over a set of local
martingale measures for $S$, through the (formal) duality
\begin{equation}
  \label{eq:duality1}
  \begin{split}
    \sup_{\theta\in\Theta}\inf_{P\in\mathcal{P}}&E_P[U(\theta\cdot S_T+B)]\\
    & =\inf_{\lambda>0}\inf_{Q\in\mathcal{M}}
    \inf_{P\in\mathcal{P}}E_P\left[V\left(\lambda \frac{dQ}{dP}\right)+\lambda
      \frac{dQ}{dP}B\right],
  \end{split}
\end{equation}
where $V$ is the conjugate of the utility function $U$ and
$\mathcal{M}$ is a set of local martingale measures for $S$.  We call
the right hand side the \emph{dual problem}.  When $B\equiv 0$ and
$\mathrm{dom}(U)=\mathbb{R}_+$, this type of duality is proved by
\cite{schied_wu05} under mild assumptions (see also
\citep{follmer_schied_weber09} and \cite{schied07}), and by
\cite{wittmuss08} for bounded endowment $B$. On the other hand, the
case of utility on the whole real line has only a few references:
\cite{follmer_gundel06} for $B\equiv 0$, \cite{owari10:_IJTAF} for the
\emph{exponential utility} with a suitably integrable $B$, and its
generalization by \cite{owari11:_AME} to the case with utility
\emph{bounded from above}.

The central question of this paper is: to what degree of generality
does (\ref{eq:duality1}) hold true?  Under the fundamental assumption
that $S$ is \emph{locally bounded}, we shall prove the duality for a
wide class of utility functions $U$ on the whole real line and
unbounded random endowments $B$. To the best of our knowledge, this is
the most general one among duality results in robust utility
maximization with $U$ being finite on the whole $\mathbb{R}$
(\cite{follmer_gundel06},
\cite{owari10:_IJTAF,owari11:_AME}). Also, this will allow
us to introduce and compute a robust version of utility indifference
prices.

It is worth emphasizing that our contribution is not only the slight
improvement of the result itself, but also to give a simple unified
proof based on a philosophy different from other related works
(\cite{schied_wu05}, \cite{schied07}, \cite{wittmuss08},
\cite{follmer_gundel06},
\cite{owari10:_IJTAF,owari11:_AME}).  A (standard) way of
showing the duality in robust utility maximization is to reduce the
robust problem to a family of \emph{subjective} problems with the help
of a minimax theorem, i.e., to interchange the order of ``$\inf$'' and
``$\sup$'' in the left hand side of (\ref{eq:duality1}) and then apply
the duality for each \emph{fixed} $P\in\mathcal{P}$. This procedure, however,
requires an additional technical assumption that the utility function
is bounded from above.

Instead, we follow a functional analytic approach in the spirit of
\cite{bellini_frittelli02} in the subjective case, which appeals to
Fenchel's general duality theorem with the help of Rockafellar's
classical result on \emph{convex integral functionals}.  To apply this
approach to our robust problem, we need a \emph{Rockafellar-type
  theorem} for a \emph{robust version} of integral functionals.  To
this end, we study the functionals of the type
\begin{align}\label{eq:RobInt1Intro}
  X\mapsto \sup_{P\in\mathcal{P}}E_P[f(\cdot, X)],
\end{align}
where $f: \Omega\times\mathbb{R}\rightarrow\mathbb{R}\cup\{+\infty\}$ is a random
convex function, and describe their conjugate functionals.  This
constitutes the heart of this article, and the duality theorem for the
robust utility maximization then follows with the idea of
\citep{bellini_frittelli02} mentioned above.

The rest of this paper is organized as follows. In
Section~\ref{sec:result}, we state the precise assumptions for the
robust utility maximization problem (\ref{eq:ProbUtil1}) and the
duality theorem (Theorem~\ref{thm:MainFinancial}) as well as some of
its consequences, including a robust version of utility indifference
prices. Then we give in Section~\ref{sec:Outline} the heuristics
behind the proof of the duality theorem and a key lemma
(Lemma~\ref{lem:Key}), which leads us to the study of functionals of
the form (\ref{eq:RobInt1Intro}). This is the subject of
Section~\ref{sec:PointWise}, which constitutes the heart of our
analysis, giving a Rockafellar-type theorem
(Theorem~\ref{thm:PointWise}) for this type of functionals.  The proof
of the duality theorem is given in Section~\ref{sec:ProofMain}.

\section{Main Result}
\label{sec:result}

\subsection{Setup}
\label{sec:SetUp}

Let $(\Omega,\mathcal{F},\mathbb{P})$ be a complete probability space, equipped with
a filtration $\mathbb{F}:=(\mathcal{F}_t)_{t\in[0,T]}$ satisfying the usual
conditions, where $T\in(0,\infty)$ is a fixed time horizon, and
$\mathcal{F}=\mathcal{F}_T$. For any probability $P\ll\mathbb{P}$, we write
$L^p(P):=L^p(\Omega,\mathcal{F},P)$, and set $L^p:=L^p(\mathbb{P})$. Also, the
expectation under $P$ is denoted by $E_P[\cdot]$, and
$E[\cdot]:=E_{\mathbb{P}}[\cdot]$. In other words, any probabilistic
notations without reference to the probability measure are preserved
for $\mathbb{P}$.

The process $S$ modeling the underlying asset prices is supposed to be
a $d$-dimensional càdlàg \emph{locally bounded} semimartingale on
$(\Omega,\mathcal{F},\mathbb{F},\mathbb{P})$. Then it is also a locally
bounded semimartingale under all $P\ll \mathbb{P}$.  A trading strategy and
its gain are modeled respectively by a $d$-dimensional predictable
$(S,\mathbb{P})$-integrable process $\theta=(\theta^1,...,\theta^d)$, and its
stochastic integral $\theta\cdot S=\int_0^\cdot\theta_sdS_s$. The next
class of integrands seems to be a natural choice of \emph{admissible
  strategies}:
\begin{equation}
  \label{eq:admissible}
  \Theta_{bb}:=\{\theta\in L(S):\,\theta_0=0,\, \theta\cdot S\text{ is uniformly bounded from below}\},
\end{equation}
where $L(S):=L(S,\mathbb{P})$ is the set of all predictable
$(S,\mathbb{P})$-integrable processes (see \cite{jacod79,jacod80} for precise
definition). Note that if $\theta\in \Theta_{bb}$, the stochastic
integral $\theta\cdot S$ is well-defined under all $P\ll\mathbb{P}$, and is
then a $P$-supermartingale by Ansel-Stricker's criterion
\citep{ansel_stricker94} (see also
\cite[Ch. 7]{delbaen_schachermayer2006:_mathem_of_arbit}). However,
the class $\Theta_{bb}$ is known to be too small in that it can not
admit an \emph{optimal strategy} even with quite ideal settings (e.g.,
$\mathcal{P}=\{\mathbb{P}\}$, $B\equiv 0$ and $S$ is a geometric Brownian motion) as
long as we work with a utility function on the entire real line (see
e.g. \cite{schachermayer2003:_super_martin_proper_of_optim_portf_proces}). Thus
the class $\Theta_{bb}$ has to be appropriately enlarged if we want to
obtain an optimal strategy, although it is not the purpose of this
article. We will introduce a possible choice of such enlargement,
after defining some more notations.

We fix also a set $\mathcal{P}$ of probabilities $P\ll \mathbb{P}$, which describes
the model uncertainty. Since $P\ll \mathbb{P}$ for all $P\in\mathcal{P}$, we can
embed the set $\mathcal{P}$ into $L^1=L^1(\mathbb{P})$ by the injection $P\mapsto
dP/d\mathbb{P}$. In other words, we identify $\mathcal{P}$ with
$\{dP/d\mathbb{P}\}_{P\in\mathcal{P}}$. Then we assume:
\begin{description}[(A3)]
\item[(A1)] $\mathcal{P}$ is convex and $\sigma(L^1,L^\infty)$-compact
  in $L^1$.
\end{description}

In this paper, we consider a utility function defined on the
\emph{entire real line}, i.e.,
$\mathrm{dom}(U)=\{x\in\mathbb{R}:U(x)\in\mathbb{R}\}=\mathbb{R}$. More specifically, we
assume;
\begin{description}[(A3)]
\item[(A2)] $U:\mathbb{R}\rightarrow\mathbb{R}$ is a strictly concave, increasing,
  and continuously differentiable function satisfying the \emph{Inada
    condition}:
  \begin{equation*}
    \lim_{x\rightarrow -\infty}U'(x)=+\infty\quad\text{and}\quad \lim_{x\rightarrow+\infty}U'(x)=0.
  \end{equation*}
\end{description}
The conjugate of $U$ is defined by
\begin{align*}
  V(y):=\sup_{x\in\mathbb{R}}(U(x)-xy),\quad y\in \mathbb{R}.
\end{align*}
The assumption (A2) implies that $V$ is a strictly convex
differentiable function with $V(y)=+\infty$ for $y<0$,
$V(0)=\sup_xU(x)$, and
\begin{equation}
  \label{eq:InadaV}
  V'(0):=\lim_{y\searrow 0}V'(y)=-\infty\text{ and }V'(\infty):=\lim_{y\nearrow+\infty}V'(y)=+\infty.
\end{equation}
In particular, $V$ is bounded from below.  Using the function $V$, we
define the $V$-divergence functional by
\begin{equation}
  \label{eq:V-Div1}
  V(\nu|P):=
  \begin{cases}
    E_P[V(d\nu/dP)]\quad&\text{if }Q\ll P,\\
    +\infty\quad&\text{otherwise},
  \end{cases}
\end{equation}
for all positive finite (resp. probability) measures $\nu\ll \mathbb{P}$
(resp. $P\ll \mathbb{P}$). We set also
\begin{align*}
  V(\nu|\mathcal{P}):=\inf_{P\in\mathcal{P}}V(\nu|P),
\end{align*}
and call this map the \emph{robust $V$-divergence}.  Note that
$(\nu,P)\mapsto V(\nu|P)$ is convex and lower semicontinuous (see
\cite[Lemma 2.7]{follmer_gundel06}), hence $\nu\mapsto V(\nu|\mathcal{P})$ is
also convex since $\mathcal{P}$ is convex by (A1).

A probability $Q\ll \mathbb{P}$ on $(\Omega,\mathcal{F})$ is called an absolutely
continuous local martingale measure if $S$ is a $Q$-local
martingale. The set of all absolutely continuous martingale measures
is denoted by $\mathcal{M}_{\mathrm{loc}}$.  For the domain of the dual
problem, we take a subset of $\mathcal{M}_{\mathrm{loc}}$:
\begin{equation}
  \label{eq:DomDual}
  \begin{split}
    \mathcal{M}_V:&=\{Q\in\mathcal{M}_{\mathrm{loc}}:\,V(\lambda Q|\mathcal{P})<\infty\text{
      for some }\lambda>0\}.
  \end{split}
\end{equation}
Note that this set is convex. To see this, suppose $V(\lambda_i
Q_i|\mathcal{P})<\infty$ for $\lambda_i>0$ and $Q_i\in \mathcal{M}_{loc}$ ($i=1,2$),
and let $\alpha\in (0,1)$. Then taking
$\gamma=\lambda_1\lambda_2/(\alpha \lambda_2+(1-\alpha)\lambda_1)$, we
have
\begin{align*}
  V(&\gamma(\alpha Q_1+(1-\alpha)Q_2)|\mathcal{P})\\
  & =V\left(
    \frac{\alpha\lambda_2}{\alpha\lambda_2+(1-\alpha)\lambda_1}\lambda_1Q_1
    +\frac{(1-\alpha)\lambda_1}{\alpha\lambda_2+(1-\alpha)\lambda_1}\lambda_2Q_2\Bigm|\mathcal{P}\right)\\
  &\leq
  \frac{\alpha\lambda_2}{\alpha\lambda_2+(1-\alpha)\lambda_1}V(\lambda_1Q_1|\mathcal{P})
  +\frac{(1-\alpha)\lambda_1}{\alpha \lambda_2+(1-\alpha)\lambda_1}V(\lambda_2Q_2|\mathcal{P})<\infty
\end{align*}
since $\nu\mapsto V(\nu|\mathcal{P})$ is convex.  We assume for $\mathcal{M}_V$:
\begin{description}[(A3)]
\item[(A3)] $\mathcal{M}_V^e:=\{Q\in\mathcal{M}_V:\,Q\sim \mathbb{P}\}\neq\emptyset$.
\end{description}
\begin{rem}\label{rem:NA1}
  This condition implies that there exists a pair $(\bar Q,\bar
  P)\in\mathcal{M}_V\times\mathcal{P}$ as well as $\bar\lambda>0$ such that $\bar
  Q\sim\bar P\sim\mathbb{P}$ and $V(\bar\lambda\bar Q|\bar P)<\infty$. In particular,
  ``$\mathcal{P}$ is equivalent to $\mathbb{P}$'': for any $A\in\mathcal{F}$,
 \begin{align}\label{eq:PequivPB}
   P(A)=0,\,\forall P\in\mathcal{P}\,\Leftrightarrow\,\mathbb{P}(A)=0.
 \end{align}
 This yields a kind of \emph{no-arbitrage} (NA) for $\mathcal{P}$: if
 $\theta\in\Theta_{bb}$ with $P(\theta\cdot S_T\geq 0)=1$ for all
 $P\in\mathcal{P}$, then $P(\theta\cdot S_T>0)=0$ for all $P\in\mathcal{P}$. However,
 each model $P\in\mathcal{P}$ may admit an arbitrage.
\end{rem}

Now we introduce another natural choice of admissible class enlarging
$\Theta_{bb}$:
\begin{equation}
  \label{eq:ThetaV}
  \Theta_{V}:=\{\theta\in L(S):\,\theta_0=0,\,\theta\cdot S\text{ is a $Q$-supermartingale, }\forall Q\in\mathcal{M}_V\}.
\end{equation}
By definition and the comment after (\ref{eq:admissible}), we have
$\Theta_{bb}\subset \Theta_V$. This is a largest natural choice of
universally defined admissible strategies from the arbitrage point of
view, and in the case without model uncertainty (i.e., $\mathcal{P}=\{\mathbb{P}\}$),
this class indeed admits an optimizer.

Finally, for the primal/dual problems to be well-defined, we need some
restriction on the endowment $B$. In this paper, we do not assume $B$
to be bounded, but instead:
\begin{description}[(A3)]
\item[(A4)] For \emph{some} $\varepsilon>0$, the family
  $\{U(-(1+\varepsilon)B^-)dP/d\mathbb{P}\}_{P\in\mathcal{P}}$ is uniformly
  integrable, and for every $P\in\mathcal{P}$, there exists some
  $\varepsilon_P>0$ such that
  \begin{align*}
    E_P[U(-\varepsilon_P B^+)]>-\infty.
  \end{align*}
\end{description}

\begin{rem}[Consequences of (A4) via Young's inequality]
  \label{rem:DirectAss}
  Note that
  \begin{equation}
    \label{eq:Young1}
    U(x)\leq V(y)+xy,\quad\forall x,y\in\mathbb{R}.
  \end{equation}
  This direct consequence of the definition of $V$ is called Young's
  inequality. From this, we have the following:
  \begin{enumerate}
  \item If we take a triplet $(\lambda, Q,P)$ with $\lambda>0$ and
    $V(\lambda Q|P)<\infty$ (then automatically $Q\ll P$), and a
    positive random variable $D$,
    \begin{align*}
      \lambda E_Q[D]\leq V(\lambda Q|P)-E_P[U(-D)].
    \end{align*}
    Taking $D=(1+\varepsilon)B^-$ and $D=\varepsilon_P B^+$, (A4)
    shows that $B^\pm\in L^1(Q)$, hence $B\in L^1(Q)$, for all $Q\in
    \mathcal{P}\cup\mathcal{M}_V$ (note that $V(Q|Q)=V(1)<\infty$ for the case of
    $\mathcal{P}$).
  \item In particular, $V(\lambda Q|\mathcal{P})+\lambda E_Q[B]$ is
    well-defined with values in $\mathbb{R}\cup\{+\infty\}$ for all $Q\in
    \mathcal{M}_V$ and $\lambda >0$ since $V$ is bounded from
    below. Therefore, the dual problem below of (\ref{eq:ProbUtil1})
    is well-defined:
    \begin{equation}
      \label{eq:DualProb1}
      \text{minimize}\quad V(\lambda Q|\mathcal{P})+\lambda E_Q[B]\,\text{over all }\lambda >0,\,Q\in\mathcal{M}_V,
    \end{equation}
    and $\inf_{\lambda>0, Q\in\mathcal{M}_V}(V(\lambda Q|\mathcal{P})+\lambda
    E_Q[B])<\infty$ by (A3).
  \item Whenever $X\in L^1(Q)$ for all $Q\in \mathcal{M}_V$, we have, for all
    $Q\in\mathcal{M}_V$,
    \begin{equation}
      \label{eq:IneqConseqYoung2}
      \begin{split}
        \inf_{P\in\mathcal{P}}E_P&[U(X+B)]\leq E_{P_Q}[U(X+B)]\\
        &\leq V(\lambda_Q Q|P_Q)+\lambda_Q(E_Q[X]+E_Q[B])<\infty,
      \end{split}
    \end{equation}
    where $P_Q\in\mathcal{P}$ and $\lambda_Q>0$ are chosen so that
    $V(\lambda_QQ|P_Q)<\infty$ (such a choice is possible by the
    definition of $\mathcal{M}_V$). This implies that the map $X\mapsto
    \inf_{P\in\mathcal{P}}E_P[U(X+B)]$ is well-defined at least on $L^\infty$
    and $\mathcal{K}:=\{\theta\cdot S_T:\, \theta\in\Theta_{bb}\}$ as a proper
    concave function. Also, taking $X=0$ in
    (\ref{eq:IneqConseqYoung2}), (A4) and the monotonicity of $U$
    imply that $\inf_{\lambda>0,Q\in\mathcal{M}_V}(V(\lambda Q|\mathcal{P})+\lambda
    E_Q[B])>-\infty$.
  \end{enumerate}

\end{rem}

\subsection{Duality Theorem and Robust Utility Indifference Valuation}
\label{sec:MainFin}

We are now in the position to state the duality theorems in its
rigorous form. We start with a basic result, then deduce some of its
consequences. In the rest of this section, (A1) -- (A4) are always in
force as the standing assumptions, and we do not cite them in each
statements. The proofs will be given in Section~\ref{sec:ProofMain}.
\begin{thm}
  \label{thm:MainFinancial}
  For any $\Theta\subset L(S)$ sandwiched by $\Theta_{bb}$ and
  $\Theta_V$, i.e., $\Theta_{bb}\subset \Theta\subset\Theta_V$,  the
  duality holds true:
  \begin{equation}
    \label{eq:ThmDuality}
    \sup_{\theta\in\Theta}\inf_{P\in\mathcal{P}}E_P[U(\theta\cdot S_T+B)]
    =\inf_{\lambda>0}\inf_{Q\in\mathcal{M}_V}\left(V(\lambda Q|\mathcal{P})+\lambda E_Q[B]\right).
  \end{equation}
  Furthermore, the infimum in the right hand side is attained by some
  $(\hat\lambda,\widehat{Q},\widehat{P})\in(0,\infty)\times\mathcal{M}_V\times\mathcal{P}$, i.e., the
  right hand side is written as $V(\hat\lambda \widehat{Q}|\widehat{P})+\hat\lambda
  E_{\widehat{Q}}[B]$.
\end{thm}

When $U(\infty):=\sup_xU(x)<\infty$ (as the exponential utility), the
same result is proved in \citep[Theorems 2.5 and 2.7]{owari11:_AME},
and our Theorem~\ref{thm:MainFinancial} slightly generalize that to
cover the case with $U(\infty)=\infty$. Although this improvement
seems rather minor, we give a unified proof based on a different
philosophy, developing a \emph{new methodology} which is the second
and a key contribution of this work. We refer to \citep{owari11:_AME}
for a review of other related results.

The reason for stating the \emph{robust duality} in a ``robust form
against the choice of $\Theta$'' is twofold. The first one concerns
the possibility of obtaining an optimal strategy as mentioned above,
from which point of view, the larger the class, the better. On the
other hand, from the practical point of view, investors living in the
real market can use only quite limited strategies (e.g., simple
strategies with bounded wealth), thus the smaller the $\Theta$, the
better for real investors, provided that the resulting maximal utility
(the left hand side of (\ref{eq:ThmDuality})) is unchanged. In terms
of the first point of view, whether the class $\Theta_V$ can admit an
optimizer for the problem (\ref{eq:ProbUtil1}) with some reasonable
generality is still open.
\begin{rem}
  Whenever $(\hat\lambda, \widehat{Q},\widehat{P})$ is a dual optimizer, i.e., the
  minimizer of the right hand side, \citep[Theorem 2.7
  (a)]{owari11:_AME} shows that $\widehat{Q}\sim\widehat{P}$. Moreover, we can
  construct a \emph{maximal} solution where the term ``maximal'' is
  used to mean the \emph{maximal support} among all solutions to the
  dual problem. However, even such a maximal solution may fail to be
  equivalent to the \emph{reference measure} $\mathbb{P}$. See
  \cite{schied07} for the concept of maximal solution and related
  information.
\end{rem}

We now give a couple of variants of
Theorem~\ref{thm:MainFinancial}. In the right hand side of
(\ref{eq:ThmDuality}), the infimum in $Q$ is taken over $\mathcal{M}_V$ which
consists of \emph{absolutely continuous} martingale measures. Then it
is natural and important to ask if we can replace $\mathcal{M}_V$ in
(\ref{eq:ThmDuality}) by $\mathcal{M}_V^e$, i.e., by elements which are
equivalent to $\mathbb{P}$. This point becomes crucial when we want to
\emph{compute} in some concrete setting the dual value function or the
indifference prices which we shall discuss below. The answer is
positive \emph{if we abandon the attainability}.
\begin{cor}
  \label{cor:Equiv}
  The infimum over $\mathcal{M}_V$ in (\ref{eq:ThmDuality}) can be replaced by
  that over $\mathcal{M}_V^e$, i.e.,
 \begin{equation}
   \label{eq:CorDualEquiv}
   \sup_{\theta\in\Theta}\inf_{P\in\mathcal{P}}E_P[U(\theta\cdot S_T+B)]
   =\inf_{\lambda>0}\inf_{Q\in\mathcal{M}^e_V}\left(V(\lambda Q|\mathcal{P})+\lambda E_Q[B]\right).
  \end{equation}
\end{cor}

We proceed to the second variant. The reader might realize that our
formulation of the (robust) utility maximization is less general than
the following form which seems more common in literature:
\begin{equation}
  \label{eq:RobUtilInitial}
  \text{maximize}\quad \inf_{P\in\mathcal{P}}E_P[U(x+\theta\cdot S_T+B)]\quad\text{over all }\theta\in\Theta,
\end{equation}
where $x\in\mathbb{R}$ is the \emph{initial capital}. If we consider a
utility on the positive half line with $B\equiv0$, the initial capital
is necessary since otherwise the problem is trivial. In our case,
however, we can embed the problem (\ref{eq:RobUtilInitial}) into
(\ref{eq:ProbUtil1}) by replacing $B$ by $x+B$. This is indeed
possible since the assumption (A4) is stable under translation by
constants: if $B$ satisfies (A4), then so does $x+B$ for any
$x\in\mathbb{R}$. We thus obtain:

\begin{cor}
  \label{cor:Initial}
  For any $\Theta$ with $\Theta_{bb}\subset \Theta\subset\Theta_V$ and
  $x\in \mathbb{R}$,
  \begin{equation}
    \label{eq:DualityX}
    \begin{split}
      \sup_{\theta\in\Theta}\inf_{P\in\mathcal{P}}&E_P[U(x+\theta\cdot S_T+B)]\\
      & =\inf_{\lambda>0}\left\{\lambda x+\inf_{Q\in\mathcal{M}^e_V}\left(V(\lambda
        Q|\mathcal{P})+\lambda E_Q[B]\right)\right\}.
    \end{split}
  \end{equation}

\end{cor}

A mathematically straightforward but practically important application
of our duality result is the indifference prices of $B$ based on a
\emph{robust preference}. We define \emph{buyer's robust utility
  indifference price} of $B$ by
\begin{equation*}
  p_b(B):=\sup\left\{p: \begin{split} 
      \textstyle{\sup_{\theta\in\Theta_{bb}}\inf_{P\in\mathcal{P}}}&E_P[U(\theta\cdot S_T-p+B)]\\
      &\geq
      \textstyle{\sup_{\theta\in\Theta_{bb}}\inf_{P\in\mathcal{P}}}E_P[U(\theta\cdot S_T)]
    \end{split}
  \right\}.
\end{equation*}
The interpretation of this ``price'' is same as in the subjective case
(i.e., $\mathcal{P}$ is a singleton, see
e.g. \cite{rouge_elkaroui00,mania_schweizer05,biagini_frittelli_Grasselli08}):
the quantity
$\sup_{\theta\in\Theta_{bb}}\inf_{P\in\mathcal{P}}E_P[U(\theta\cdot S_T-p+B)]$
represents the maximal admissible robust utility \emph{with buying the
  claim $B$ at the price $p$}, while
$\sup_{\theta\in\Theta_{bb}}\inf_{P\in\mathcal{P}}E_P[U(\theta\cdot S_T)]$ is
that \emph{without buying the claim}; thus $p_b(B)$ is understood the
\emph{maximal acceptable price} of $B$ for the buyer whose preference
is determined by the robust utility functional $X\mapsto
\inf_{P\in\mathcal{P}}E_P[U(X)]$. Seller's price is defined in a symmetric
way.  Given Corollary~\ref{cor:Initial}, the following expression is
immediate.
\begin{cor}
  \label{cor:IndiffPrice}
  Buyer's robust utility indifference price $p_b(B)$ of $B$ is
  expressed as:
  \begin{equation}
    \label{eq:IndiffBuy}
    p_b(B)=\inf_{Q\in\mathcal{M}^e_V}(E_Q[B]+\gamma(Q)),
  \end{equation}
  where
  \begin{equation}
    \label{eq:RiskMeasPenal}
    \gamma(Q):=\inf_{\lambda>0}\frac{1}{\lambda}\Bigl(V(\lambda Q|\mathcal{P})
      -\inf_{\lambda'>0}\inf_{Q'\in\mathcal{M}^e_V}V(\lambda' Q'|\mathcal{P})\Bigr).
  \end{equation}
\end{cor}
\begin{proof}[Proof of Corollary~\ref{cor:IndiffPrice}]
  By Corollary~\ref{cor:Initial}, we have
  \begin{align*}
    \sup_{\theta\in\Theta_{bb}}\inf_{P\in\mathcal{P}}E_P[U(\theta\cdot
    S_T-p+B)]%
    &=\inf_{\lambda>0}\inf_{Q\in\mathcal{M}^e_V}(V(\lambda Q|\mathcal{P})-\lambda p+\lambda E_Q[B]),\\
    \sup_{\theta\in\Theta_{bb}}\inf_{P\in\mathcal{P}}E_P[U(\theta\cdot
    S_T)]&=\inf_{\lambda>0}\inf_{Q\in\mathcal{M}^e_V} V(\lambda Q|\mathcal{P})=:v_0.
  \end{align*}
  Hence,
  \begin{align*}
    &\sup_{\theta\in\Theta_{bb}}\inf_{P\in\mathcal{P}}E_P[U(\theta\cdot
    S_T-p+B)]\geq \sup_{\theta\in\Theta_{bb}}\inf_{P\in\mathcal{P}}E_P[U(\theta\cdot S_T)]\\
    \Leftrightarrow&
    \inf_{\lambda>0}\inf_{Q\in\mathcal{M}^e_V}(V(\lambda Q|\mathcal{P})-\lambda p+\lambda E_Q[B])\geq v_0\\
    \Leftrightarrow&\inf_{Q\in\mathcal{M}^e_V}(V(\lambda Q|\mathcal{P})+\lambda E_Q[B])-\lambda p\geq v_0,\quad\forall \lambda>0\\
    \Leftrightarrow& \inf_{Q\in\mathcal{M}^e_V}(E_Q[B]+\gamma(Q))\\
&\qquad    =\inf_{Q\in\mathcal{M}^e_V}\left(E_Q[B]+\frac{1}{\lambda}\left(V(\lambda
        Q|\mathcal{P})-v_0\right)\right)\geq p,\quad\forall \lambda>0.
  \end{align*}
  Therefore, we have (\ref{eq:IndiffBuy}).
\end{proof}

\subsection{Heuristics and a key Lemma}
\label{sec:Outline}

To motivate ourselves, and to make it clear what the difficulty and
our contribution are, we give here a functional analytic view on the
duality theorem. We begin by recalling a classical theorem in
functional analysis, called Fenchel's duality theorem. Let
$f:E\rightarrow \mathbb{R}\cup\{+\infty\}$
(resp. $g:E\rightarrow\{-\infty\}$) be a proper convex (resp. concave)
function defined on a topological vector space $E$. Then Fenchel's
theorem (Rockafellar's version \citep{rockafellar66:_fenchel}) states
that \emph{if either $f$ or $g$ is continuous at some $x_0\in
  \mathrm{dom}f\cap\mathrm{dom}g$, then}
\begin{equation}
  \label{eq:Fenchel1}
  \sup_{x\in E}(g(x)-f(x))=\min_{y\in E^*}(f^*(x)-g_*(y)),
\end{equation}
where $E^*$ denotes the (topological) dual of $E$, and $f^*$
(resp. $g_*$) is the convex (resp. concave) conjugate of $f$
(resp. $g$) defined by
\begin{align*}
  f^*(x)&:=\sup_{x\in E}(\langle x,y\rangle-f(x)),\quad \forall y\in E^*\\
  g_*(y)&:=\inf_{x\in E}(\langle x,y\rangle -g(x)),\quad\forall y\in
  E^*.
\end{align*}
If we take a convex cone $C\subset E$, and set $f(x)=\delta_C(x):=0$
(resp. $:=+\infty$) if $x\in C$ (resp. $x\not\in C$), which is
obviously convex, then its conjugate is given by
\begin{align*}
  \delta_C^*(y)=\sup_{x\in C}\langle x,y\rangle=\delta_{C^\circ}(y),
\end{align*}
where $C^\circ$ is the \emph{polar cone} defined by
\begin{align*}
  C^\circ:=\{y\in E^*:\, \langle x,y\rangle\leq 0,\, \forall x\in C\}.
\end{align*}
Then if $g$ is continuous at some point $x_0\in C$, we have
\begin{equation}
  \label{eq:FenchelCone}
  \sup_{x\in C}g(x)=\sup_{x\in E}(g(x)-\delta_C(x))=\min_{y\in C^\circ}-g_*(y).
\end{equation}

This is a well-established duality argument in \emph{abstract
  functional analysis}. We now translate our \emph{specific} problem
into this framework. We take $E=L^\infty$, then its dual is $ba$ with
$\langle X,\nu\rangle=\int_\Omega Xd\nu=:\nu(X)$ (see
Definition~\ref{dfn:ba} and Lemma~\ref{lem:ba}). Also, define
\begin{equation}
  \label{eq:SetC}
  \mathcal{C}:=\{X\in L^\infty:\,X\leq \theta\cdot S_T,\,\exists \theta\in\Theta_{bb}\}.
\end{equation}
This is the set of all super-replicable claims with zero initial
costs, and it is well-known (see
e.g. \citep{delbaen_schachermayer2006:_mathem_of_arbit}) in
mathematical finance that $\mathcal{C}$ is a convex cone containing
$L^\infty_-$, and every \emph{$\sigma$-additive} element of
$\mathcal{C}^\circ$ is a positive multiple of some local martingale measure,
that is,
\begin{equation}
  \label{eq:PosMultiple}
  \mathcal{C}^\circ\cap ba^\sigma=\{\lambda Q:\, \lambda\geq 0, \, Q\in\mathcal{M}_\mathrm{loc}\}=:\mathrm{cone}\mathcal{M}_\mathrm{loc}.
\end{equation}
Here $ba^\sigma$ denotes the set of all \emph{$\sigma$-additive}
elements of $ba$ (i.e., finite signed measures $\nu\ll \mathbb{P}$).

For the concave function $g$ and its conjugate, we take
\begin{align}
  \label{eq:UtilFunctionB}
  u_{B,\mathcal{P}}(X)&:=\inf_{P\in\mathcal{P}}E_P[U(X+B)], \quad X\in L^\infty,\\
  \label{eq:ConjVB}
  v_{B,\mathcal{P}}(\nu)&:=\sup_{X\in L^\infty}(u_{B,\mathcal{P}}(X)-\nu(X))=-(u_{B,\mathcal{P}})_*(\nu),\quad\nu\in ba.
\end{align}
Note that $u_{B,\mathcal{P}}$ is well-defined proper concave function on
$L^\infty$ due to Item 3 of Remark~\ref{rem:DirectAss} above. Now if
$u_{B,\mathcal{P}}$ is continuous, the above argument shows:
\begin{equation}
  \label{eq:FenchelOutline2}
  \sup_{X\in\mathcal{C}}u_{B,\mathcal{P}}(X)=\min_{\nu\in \mathcal{C}^\circ}v_{B,\mathcal{P}}(\nu)=\min_{\nu\in \mathcal{C}^\circ\cap\mathrm{dom }\, v_{B,\mathcal{P}}}v_{B,\mathcal{P}}(\nu).
\end{equation}
Thus, if we have the continuity of $u_{B,\mathcal{P}}$ and the complete
description of $v_{B,\mathcal{P}}$, everything else is cleared by the
classical duality arguments via Fenchel's theorem. The key is
therefore:
\begin{lem}[Key Lemma]
  \label{lem:Key}
  Assume (A1) -- (A4). Then $u_{B,\mathcal{P}}$ is finite and norm continuous
  on the whole $L^\infty$, and
  \begin{align}\label{eq:ConjVRepKeyLemma1}
    v_{B,\mathcal{P}}(\nu)=
    \begin{cases}
        V(\nu|\mathcal{P})+\nu(B)&\quad\text{if $\nu\in ba_+^\sigma$ and }V(\nu|\mathcal{P})<\infty,\\
        +\infty&\quad\text{otherwise,}
      \end{cases}
    \end{align}
  where $\nu(B):=\int_\Omega Bd\nu$.
\end{lem}
As soon as we obtain this key lemma, the above arguments shows the
following \emph{abstract duality}:
\begin{equation}
  \label{eq:AbstractDuality}
  \sup_{X\in\mathcal{C}}u_{B,\mathcal{P}}(X)=\min_{\lambda\geq0, Q\in\mathcal{M}_V}(V(\lambda Q|\mathcal{P})+\lambda E_Q[B]).
\end{equation}
Indeed, the representation (\ref{eq:ConjVRepKeyLemma1}) and
(\ref{eq:PosMultiple}) imply $\mathcal{C}^\circ \cap\mathrm{dom}\,
v_{B,\mathcal{P}}\subset \mathrm{cone}\mathcal{M}_V$, and (\ref{eq:FenchelOutline2})
shows (\ref{eq:AbstractDuality}). Then some slight adjustments will
finish the proof of the duality (\ref{eq:ThmDuality}).

If we look at this lemma from more abstract point of view, we come up
to the convex functionals of the type:
\begin{equation}
  \label{eq:RobIntFunc}
  \mathcal{I}_{f,\mathcal{P}}(X):=\sup_{P\in\mathcal{P}}E_P[f(\cdot, X)],\quad X\in L^\infty,
\end{equation}
where $f:\Omega\times\mathbb{R}\rightarrow \mathbb{R}\cup\{+\infty\}$ is a
\emph{random convex function}. In fact, if we take
$f(\omega,x)=-U(-x+B(\omega))$, then $u_{B,\mathcal{P}}(X)=-\mathcal{I}_{f,\mathcal{P}}(-X)$,
and $v_{B,\mathcal{P}}(\nu)=\mathcal{I}_{f,\mathcal{P}}^*(\nu)$. When $\mathcal{P}$ consists of a
single element (say $\{\mathbb{P}\}$), the corresponding functional $X\mapsto
E[f(\cdot, X)]$ is called the \emph{convex integral functional}, and
its basic properties including the expression of the conjugate
functional are obtained by R. T. Rockafellar in
\citep{rockafellar68:_integ_whichI,rockafellar71,rockafellar76:_integ_funct_normal}.
The core of our analysis is thus to extend these classical results to
the \emph{robust version} of integral functionals of the form
(\ref{eq:RobIntFunc}).

\section{Analysis of Robust Convex Functionals}
\label{sec:PointWise}

We now proceed to the key part of the paper, namely the analysis of
the robust convex functionals formally defined by
(\ref{eq:RobIntFunc}). We start by introducing some terminologies.

\begin{dfn}[Normal Convex Integrands]
  \label{dfn:NormalIntegrands}
  A map $f:\Omega\times \mathbb{R}\mapsto \mathbb{R}\cup\{+\infty\}$ is said to be a
  \emph{normal convex integrand} if:
  \begin{enumerate}
  \item $f$ is $\mathcal{F}\otimes\mathcal{B}(\mathbb{R})$-measurable;
  \item $x\mapsto f(\omega,x)$ is a lower semicontinuous proper convex
    function for a.e. $\omega$.
  \end{enumerate}
\end{dfn}

Several equivalent definitions of normality are found in
\cite{rockafellar_wets98}. An immediate but important consequence of
the normality in this sense is that the map $\omega\mapsto
f(\omega,X(\omega))$ is $\mathcal{F}$-measurable, hence the following
definition makes sense:
\begin{equation}
  \label{eq:IntFuncClassical}
  I_f(X):=E[f(\cdot,X)],\quad X\in L^\infty.
\end{equation}
Note that the normality of $f$ already implies that the
``$\omega$-wise'' conjugate defined below is also normal:
\begin{equation}
  \label{eq:OutlineConjNormal1}
  f^*(\omega,y):=\sup_{x\in\mathbb{R}}(xy-f(\omega,x)),\quad\forall y\in \mathbb{R},
\end{equation}
Therefore, the integral functional $I_{f^*}$ on $L^1$ also makes sense
(provided some integrability condition).

As we are considering the integral functionals on $L^\infty$, we need
the description of the dual of $L^\infty$, that is the space $ba$,
which already appeared in Section~\ref{sec:Outline}. For the article
to be self-contained, we recall the definition and some of basic
properties. For the proofs and more details, we refer the reader to
\cite{dunford58:_linear_operat_I}, \cite{hewitt_stromberg75:_real} and
\cite{rao_rao83:_theor}.
\begin{dfn}\label{dfn:ba}
  $ba:=ba(\Omega,\mathcal{F},\mathbb{P})$ is the set of all bounded finitely
  additive signed measures absolutely continuous w.r.t. $\mathbb{P}$, i.e.,
  $\nu\in ba$ if and only if $\nu$ is a real valued function on
  $\mathcal{F}$ such that (1) $\sup_{A\in\mathcal{F}}|\nu(A)|<\infty$,
  (2) for every $A\in\mathcal{F}$, $\mathbb{P}(A)=0$ implies $\nu(A)=0$, (3)
  if $A,B\in\mathcal{F}$ and $A\cap B=\emptyset$, then $\nu(A\cup
  B)=\nu(A)+\nu(B)$.  Also, $ba_+$ (resp. $ba^\sigma$) denotes the set
  of positive (resp. $\sigma$-additive) elements of $ba$, and set
  $ba_+^\sigma:= ba_+\cap ba^\sigma$.
\end{dfn}

\begin{lem}
  \label{lem:ba}
  The following assertions hold:
  \begin{enumerate}
  \item (Jordan decomposition). Every $\nu\in ba$ admits a unique
    decomposition $\nu=\nu_+-\nu_-$, where $\nu_{\pm}\in ba_+$. Thus,
    $|\nu|=\nu_++\nu_-\in ba_+$ is well-defined.
  \item $ba$ is a Banach space equipped with the total variation norm
    $\|\nu\|=|\nu|(\Omega)$, and $ba\simeq (L^\infty)^*$ with $\langle
    X,\nu\rangle=\int_\Omega Xd\nu=:\nu(X)$. Also, $L^1\simeq ba^\sigma$ is a
    closed subspace of $ba$.
  \item (Yosida-Hewitt decomposition) Every $\nu\in ba$ admits a
    unique decomposition $\nu=\nu_r+\nu_s$, where $\nu_r\in ba^\sigma$
    and $\nu_s$ is purely finitely additive, i.e., for every $\mu\in
    ba^\sigma_+$, $0\leq \mu\leq |\nu_s|$ implies $\mu=0$.
  \item $\nu\in ba$ is purely finitely additive if and only if there
    exists an increasing sequence $(A_n)\subset \mathcal{F}$ such that
    $\mathbb{P}(A_n)\nearrow 1$ and $|\nu|(A_n)=0$ for all $n$.
  \end{enumerate}
\end{lem}

With these preparations, we now recall the classical Rockafellar
theorem.
\begin{thm}[{\citep[Theorem 1, Corollary 2A]{rockafellar71}}]
  \label{thm:RockafellarClassical}
  Let $f$ and $f^*$ be as above, and assume:
  \begin{align}\label{eq:R1}
    &\text{ there exists some $X\in L^\infty$ such that
      $f(\cdot,X)^+\in L^1$;}\\\label{eq:R2}%
    &\text{ there exists some $Y\in L^1$ such that $f^*(\cdot,Y)^+\in
      L^1$.}
\end{align}
Then $I_{f}$ (resp. $I_{f^*}$) is well-defined as a lower
semicontinuous proper convex functional on $L^\infty$ (resp. $L^1$),
and for any $\nu\in ba$,
  \begin{equation}
    \label{eq:ThmRock1}
    \sup_{X\in L^\infty}(\nu(X)-I_{f}(X))=I_{f^*}(d\nu_r/d\mathbb{P})+\sup_{X\in\mathrm{dom}(I_{f})}\nu_s(X),
  \end{equation}
  where $\nu_r$ (resp. $\nu_s$) is the regular (resp. singular) part
  of $\nu$ in the Yosida-Hewitt decomposition. In particular, if
  $\mathrm{dom}(I_{f})=L^\infty$, then $I_{f}$ is norm continuous on
  the whole $L^\infty$ and
  \begin{equation}
    \label{eq:ThmRock2}
    \sup_{X\in L^\infty}(\nu(X)-I_{f}(X))=
    \begin{cases}
      I_{f^*}(d\nu/d\mathbb{P})&\quad\text{if $\nu$ is $\sigma$-additive,}\\
      +\infty&\quad\text{otherwise.}
    \end{cases}
  \end{equation}
\end{thm}

\begin{rem}
  All the assertions of the theorem remain valid if we replace $\mathbb{P}$
  by any $P\sim \mathbb{P}$ with some obvious modifications, e.g., replacing
  $L^1$ by $L^1(P)$, $I_f$ by $I_{f,P}(X):=E_P[f(\cdot,X)]$, and
  $d\nu/d\mathbb{P}$ by $d\nu/dP$.
\end{rem}

\subsection{Robust Version of Integral Functionals}
\label{sec:RobIntFunc}

We now proceed to the robust situation.  Let $\mathcal{P}$ be as in
Section~\ref{sec:SetUp}, satisfying the assumption (A1) of convexity
and weak compactness, and $f$ a normal convex integrand. From a
technical reason, we assume that $f$ is \emph{finite valued}, i.e.,
\begin{equation}
  \label{eq:AsFFinite}
  \mathbb{P}(f(\cdot, x)<\infty,\,\forall x\in\mathbb{R})=1.
\end{equation}
In particular, $x\mapsto f(\omega,x)$ is continuous.  We then define
the robust version of integral functional associated to $f$:
\begin{equation}
  \label{eq:SupIntFunc1}
  \mathcal{I}_{f,\mathcal{P}}(X):=\sup_{P\in\mathcal{P}}E_P[f(\cdot,X)]=\sup_{P\in\mathcal{P}}\int_\Omega f(\omega,X(\omega))P(d\omega),
  \quad \forall X\in L^\infty,
\end{equation}
and its conjugate defined on $ba$:
\begin{equation}
  \label{eq:ConvConjRob1}
  (\mathcal{I}_{f,\mathcal{P}})^*(\nu):=\sup_{X\in L^\infty}(\nu(X)-\mathcal{I}_{f,\mathcal{P}}(X)),\quad \forall \nu\in ba=ba(\Omega,\mathcal{F},\mathbb{P}).
\end{equation}
In what follows, we investigate the regularity properties of
$\mathcal{I}_{f,\mathcal{P}}$ and the description of the conjugate $(\mathcal{I}_{f,\mathcal{P}})^*$
in the spirit of Rockafellar's theorem.

The functional $\mathcal{I}_{f,\mathcal{P}}$ is the pointwise supremum over $P\in\mathcal{P}$
of the $P$-integral functionals $X\mapsto I_{f,P}(X):=E_P[f(\cdot,X)]$
restricted to $L^\infty=L^\infty(\mathbb{P})$. But there is an alternative
point of view that $\mathcal{I}_{f,\mathcal{P}}$ is a pointwise supremum over $\mathcal{P}$ of
``$\mathbb{P}$-integral functionals'' associated to the integrands
$(dP/d\mathbb{P})f$. Note that these new integrands are normal whenever $f$
is, and we can naturally deal with all $I_{f,P}$ on the same domain
$L^\infty$. This simple point of view leads us to another type of
conjugate: $\tilde f(\omega, y,z):=(zf)^*(\omega, y)$, that is,
\begin{equation}
  \label{eq:ConjTilde1}
  \tilde f(\omega, y,z):=\sup_{x\in\mathbb{R}}(xy-zf(\omega, x)),\quad \omega\in\Omega,\,y\in\mathbb{R}, \,z\geq0.
\end{equation}
By definition, we have
\begin{equation}
  \label{eq:IneqYoungFtilde}
  xy\leq zf(\omega, x)+\tilde f(\omega,y,z),\quad\forall \omega\in\Omega,\,\forall x,y \in \mathbb{R},\, \forall z\geq 0.
\end{equation}
The next lemma is elementary.
\begin{lem}
  \label{lem:ConjTilde1}
  $\tilde f$ is a normal convex integrand on
  $\Omega\times\mathbb{R}\times\mathbb{R}_+$, i.e., $\tilde f$ is
  $\mathcal{F}\otimes\mathcal{B}(\mathbb{R})\otimes\mathcal{B}(\mathbb{R}_+)$-measurable, and $(y,z)\mapsto
  \tilde f(\omega,y,z)$ is a lower semicontinuous proper convex
  function on $\mathbb{R}\times\mathbb{R}_+$. Moreover, $\tilde f$ is explicitly
  written as:
  \begin{align}\label{eq:ConjTilde2}
    \tilde f(y,z)=
    \begin{cases}
      0&\text{ if }y=z=0,\\
      +\infty&\text{ if }y\neq0, z=0,\\
      zf^*(\cdot,y/z)&\text{ if }z>0.
    \end{cases}
  \end{align}

\end{lem}
\begin{proof}
  Since $\tilde f(\cdot, y,z)\geq -zf(\cdot, 0)>-\infty$ and $\tilde
  f(\cdot, 0,0)=0$, $\tilde f$ is proper.  The lower semicontinuity
  and the convexity are consequences of the fact that $\tilde f$ is a
  point-wise supremum of linear functions $(y,z)\mapsto xy-zf(x)$ when
  $x$ runs through all reals. As $x\mapsto f(\cdot, x)$ is lower
  semicontinuous (actually continuous), the supremum over reals can be
  replaced by that over \emph{rationals}, which shows the
  measurability of $\tilde f$.  Finally, (\ref{eq:ConjTilde2}) is
  verified by direct computation.
\end{proof}

Using $\tilde f$, we introduce another type of \emph{integral
  functional} which plays the role of $I_{f^*}$ in the classical
case:
\begin{equation}
  \label{eq:ConjJ}
  \mathcal{J}_{\tilde f,\mathcal{P}}(Y):=\inf_{P\in\mathcal{P}}E[\tilde f(\cdot,Y,dP/d\mathbb{P})],\quad Y\in L^1.
\end{equation}

We first verify that the functionals $\mathcal{I}_{f,\mathcal{P}}$ and $\mathcal{J}_{\tilde
  f,\mathcal{P}}$ are indeed well-defined, under natural integrability
conditions corresponding to (\ref{eq:R1}) and (\ref{eq:R2}).
\begin{prop}\label{prop:SupIntWellDef}
  Assume (\ref{eq:AsFFinite}) and that
  \begin{align}
    \label{eq:IntegRobust1}
    &\text{ for some $X\in L^\infty$, }\sup_{P\in\mathcal{P}}E_P[f(\cdot, X)^+]<\infty;\\
    \label{eq:IntegRobust2}
    &\text{ for any $P\in\mathcal{P}$, there exists $Y\in L^1$ such that }\tilde f(\cdot,Y,dP/d\mathbb{P})^+\in L^1.
  \end{align}
  Then we have the following.
  \begin{description}[(a)]
  \item[(a)] $\mathcal{I}_{f,\mathcal{P}}$ is well-defined as a lower semicontinuous
    proper convex functional on $L^\infty$;

  \item[(b)] $\mathcal{J}_{\tilde f,\mathcal{P}}$ is well-defined as a proper convex
    functional on $L^1$;
  \item[(c)] for all $X\in L^\infty$ and $Y\in L^1$,
    \begin{equation}
      \label{eq:YoungIneqRob}
      E[XY]\leq \mathcal{I}_{f,\mathcal{P}}(X)+\mathcal{J}_{\tilde f,\mathcal{P}}(Y).
    \end{equation}
  \end{description}
\end{prop}

\begin{proof}
  (a). Fix $P\in\mathcal{P}$, and take $Y\in L^1$ as in
  (\ref{eq:IntegRobust2}).  By (\ref{eq:IneqYoungFtilde}) , we have
  \begin{align}\label{eq:ProofLemYoungConj1A}
    \frac{dP}{d\mathbb{P}}f(\cdot,X)\geq XY-\tilde f\left(\cdot,
      Y,\frac{dP}{d\mathbb{P}}\right)^+,\quad \forall X\in L^\infty.
  \end{align}
  Since the right hand side is integrable, we see that
  $E_P[f(\cdot,X)]>-\infty$ for all $X\in L^\infty$, while
  $E_P[f(\cdot, X)]<\infty$ for some $X\in L^\infty$ by
  (\ref{eq:IntegRobust1}).  Thus, the functional $X\mapsto
  E_P[f(\cdot,X)]$ is a proper convex functional on $L^\infty$, where
  the convexity is clear from the convexity of $f$.  Also, if
  $(X_n)\subset L^\infty$ is any norm convergent sequence with the
  limit $X$, then $X^nY\rightarrow XY$, a.s. and in $L^1$. Then
  (\ref{eq:ProofLemYoungConj1A}) allows us to use Fatou's lemma to
  conclude
  \begin{align*}
    E_P[f(\cdot,X)]\leq \liminf_{n\rightarrow\infty}E_P[f(\cdot, X_n)],
  \end{align*}
  hence we have that $X\mapsto E_P[f(\cdot, X)]$ is a lower
  semicontinuous proper convex functional on $L^\infty$. Since this
  holds for all $P\in\mathcal{P}$, $X\mapsto \sup_{P\in\mathcal{P}}E_P[f(\cdot,X)]$ is
  again a lower semicontinuous convex functional as a point-wise
  supremum of such functionals, and $I_{f,\mathcal{P}}(X)>-\infty$ for all
  $X\in L^\infty$. Finally, $\mathcal{I}_{f,\mathcal{P}}(X)<\infty$ for some $X\in
  L^\infty$ again by (\ref{eq:IntegRobust1}), hence $\mathcal{I}_{f,\mathcal{P}}$ is
  proper.

  (b), (c). Take $X$ as in (\ref{eq:IntegRobust1}). Then $f(\cdot,
  X)\in L^1(P)$ for all $\mathcal{P}$, and by (\ref{eq:IneqYoungFtilde}),
  \begin{align*}
    \tilde f(\cdot,Y,dP/d\mathbb{P})\geq
    XY-\frac{dP}{d\mathbb{P}}f(\cdot,X),\quad\forall Y\in L^1,\,\forall
    P\in\mathcal{P}.
  \end{align*}
  Therefore, $\mathcal{J}_{f,\mathcal{P}}$ is proper since
  \begin{equation}\label{eq:ProofYoungconj2}
    \begin{split}
      \inf_{P\in\mathcal{P}}E[\tilde f(\cdot,Y,dP/d\mathbb{P})]&\geq
      \inf_{P\in\mathcal{P}}\left(E[XY]-E_P[f(\cdot,X)]\right)\\
      &=E[XY]-\sup_{P\in\mathcal{P}}E_P[f(\cdot,X)]>\infty,\,\forall Y\in
      L^1.
    \end{split}
  \end{equation}
  The convexity of $\mathcal{J}_{\tilde f,\mathcal{P}}$ follows from that of
  $(y,z)\mapsto \tilde f(\cdot,y,z)$ (Lemma~\ref{lem:ConjTilde1}) and
  of $\mathcal{P}$, and we have (b).  As for (c), we may assume
  $\mathcal{I}_{f,\mathcal{P}}(X)<\infty$ and $\mathcal{J}_{\tilde f,\mathcal{P}}(Y)<\infty$, since
  otherwise the assertion is trivial. But this case is already proved
  by (\ref{eq:ProofYoungconj2}).
\end{proof}

\begin{rem}\label{rem:IntegFuncAss}
  The condition (\ref{eq:IntegRobust2}) is equivalent to:
\begin{equation}
  \label{eq:B2equiv}
  \forall P\in\mathcal{P}, \, \exists \widetilde{Y} \in L^1(P)\text{ such that } f^*(\cdot,\widetilde{Y})\in L^1(P).
\end{equation}
Indeed, if $\tilde f(\cdot, Y,dP/d\mathbb{P})^+\in L^1$, then $Y=0$ on
$\{\frac{dP}{d\mathbb{P}}=0\}$ by (\ref{eq:ConjTilde2}) and $\tilde f(\cdot,
Y,dP/d\mathbb{P})^+ \linebreak
=1_{\{\frac{dP}{d\mathbb{P}}>0\}}\frac{dP}{d\mathbb{P}}f^*(\cdot,
Y/(dP/d\mathbb{P}))^+$. Therefore,
$\widetilde{Y}:=1_{\{\frac{dP}{d\mathbb{P}}>0\}}Y/(dP/d\mathbb{P})$, which is $P$-integrable
since $Y$ is $\mathbb{P}$-integrable, satisfies the condition. Conversely, if
$\widetilde{Y}\in L^1(P)$ satisfies $f^*(\cdot, \widetilde{Y})^+\in L^1(P)$, we set $Y=\widetilde{Y}
dP/d\mathbb{P}$. Then we have by (\ref{eq:ConjTilde2}) that $\tilde f(\cdot,
Y,dP/d\mathbb{P})^+=\frac{dP}{d\mathbb{P}}f^*(\cdot, \widetilde{Y})^+\in L^1$. In particular,
(\ref{eq:IntegRobust2}) coincides with (\ref{eq:R2}) of
Theorem~\ref{thm:RockafellarClassical} when $\mathcal{P}$ is a singleton.
\end{rem}

\subsection{A Robust Version of the Rockafellar Theorem}
\label{sec:IntFuncConj}

We have arrived at the main theorem of this section, which is also the
heart of the whole paper. Recall that $\mathcal{I}_{f,\mathcal{P}}$ and $\mathcal{J}_{\tilde
  f,\mathcal{P}}$ are defined respectively by (\ref{eq:ConvConjRob1}) and
(\ref{eq:ConjJ}), and the normal convex integrand $f$ is assumed to be
finite valued, thus in particular, $x\mapsto f(\omega,x)$ is
continuous for a.e. $\omega$.

To obtain a nice description of the conjugate $(\mathcal{I}_{f,\mathcal{P}})^*$, the
integrability assumption (\ref{eq:IntegRobust1}) is not enough, and we
need a slightly stronger assumption. Let 
\begin{align}
  \label{eq:IntFuncDomRob}
  \mathcal{D}&:=\{X\in L^\infty:\,\{f(\cdot, X)^+dP/d\mathbb{P}\}_{P\in\mathcal{P}}\text{ is uniformly integrable}\}.
\end{align}
If $X\in \mathcal{D}$, then
$\sup_{P\in\mathcal{P}}E_P[f(\cdot,X)]\leq\sup_{P\in\mathcal{P}}E_P[f(\cdot,
X)^+]<\infty$, hence $X\in \mathrm{dom}(\mathcal{I}_{f,\mathcal{P}})$, i.e.,
$\mathcal{D}\subset\mathrm{dom}(\mathcal{I}_{f,\mathcal{P}})$.  The proof of the next theorem
will be given in Section~\ref{sec:ProofIntRepMain}.
\begin{thm}
  \label{thm:PointWise}
  Assume (\ref{eq:AsFFinite}), (\ref{eq:IntegRobust2}), (A1) (i.e.,
  $\mathcal{P}$ is convex and weakly compact) and that $\mathcal{D}\neq\emptyset$.
  Then for all $\nu\in ba$ with the Yosida-Hewitt decomposition
  $\nu=\nu_r+\nu_s$,
    \begin{equation}
      \label{eq:EstConj}
      \begin{split}
        \mathcal{J}_{\tilde f,\mathcal{P}}\left(\frac{d\nu_r}{d\mathbb{P}}\right) +\sup_{X\in
          \mathcal{D}}\nu_s(X) &\leq (\mathcal{I}_{f,\mathcal{P}})^*(\nu)\\
        &\leq \mathcal{J}_{\tilde f,\mathcal{P}}\left(\frac{d\nu_r}{d\mathbb{P}}\right)
        +\sup_{X\in\mathrm{dom}(\mathcal{I}_{f,\mathcal{P}})}\nu_s(X).
      \end{split}
    \end{equation}

\end{thm}

\begin{cor}[Restriction to $ba^\sigma\simeq L^1$]
  \label{cor:RobRockL1}
  For any $\nu\in ba^\sigma$, we have
  \begin{equation}
    \label{eq:RepIntRobL1}
    (\mathcal{I}_{f,\mathcal{P}})^*(\nu)=\mathcal{J}_{\tilde f,\mathcal{P}}(d\nu/d\mathbb{P})=\inf_{P\in\mathcal{P}}E[\tilde f(\cdot, d\nu/d\mathbb{P},dP/d\mathbb{P})].
  \end{equation}
  In particular, $\mathcal{J}_{\tilde f,\mathcal{P}}$ is also lower semicontinuous.
\end{cor}
\begin{proof}
  The first assertion is immediate from (\ref{eq:EstConj}) since the
  second terms in the left and right hand sides are zero if
  $\nu_s=0$. The lower semicontinuity follows from a general fact: if
  $\langle E,E'\rangle$ is a dual pair (see
  \cite{aliprantis_border06} for definition), and $\phi$ is a convex
  function on $E$, the conjugate $\phi^*$ on $E'$ is
  $\sigma(E',E)$-lower semicontinuous. Here $E'$ does not have to be
  the \emph{topological dual} of $E$ under the original topology.
\end{proof}

When $\mathcal{P}$ is a singleton, two sets $\mathcal{D}$ and
$\mathrm{dom}(\mathcal{I}_{f,\{P\}})$ coincide, hence the two inequalities in
(\ref{eq:EstConj}) actually holds as a single equality, which is
nothing other than the Rockafellar theorem
(Theorem~\ref{thm:RockafellarClassical}).  In the general case,
however, the inclusion $\mathcal{D}\subset \mathrm{dom}(\mathcal{I}_{f,\mathcal{P}})$ can be
strict (see Examples~\ref{ex:1} and \ref{ex:2}).

To illustrate the situation, we give an alternative form of
integrability conditions. Define
\begin{align*}
  \mathcal{L}^1(\mathcal{P})&:=\left\{X\in L^0:\, \|X\|_{1,\mathcal{P}}:=\sup_{P\in\mathcal{P}}E_P[|X|]<\infty\right\}.
\end{align*}
This is an ``$L^1$-type space'' under the \emph{sublinear expectation}
$X\mapsto\sup_{P\in\mathcal{P}}E_P[X]$. If we introduce an equivalence
relation by $X\sim_{\mathcal{P}} Y$ iff $X=Y$, $\mathcal{P}$-\emph{quasi surely}
($\Leftrightarrow$ $P$-a.s. for all $P\in\mathcal{P}$), the resulting quotient
space $L^1(\mathcal{P}):=\mathcal{L}^1(\mathcal{P})/\!\sim_\mathcal{P}$ is indeed a Banach space with
the norm $\|\cdot\|_{1,\mathcal{P}}$. Then the integrability condition
(\ref{eq:IntegRobust1}) is equivalent to saying that there exists some
$X\in L^\infty$ such that $f(\cdot, X)^+\in L^1(\mathcal{P})$. If $\mathcal{P}$ is a
singleton, the $L^1$-type space $L^1(\mathcal{P})$ is nothing but the usual
$L^1$ space, and the condition (\ref{eq:IntegRobust1}) coincides with
(\ref{eq:R1}) in the Rockafellar theorem.

Under the sublinear expectation, there is another natural $L^1$-type
space:
\begin{align*}
  \mathcal{L}_u^1(\mathcal{P}):=\left\{X\in L^0:\, \lim_{N\rightarrow \infty}\sup_{P\in\mathcal{P}}E_P[|X|1_{\{|X|\geq N\}}]=0\right\},
\end{align*}
and we set $L^1_u(\mathcal{P}):=\mathcal{L}_u^1(\mathcal{P})/\!\sim_\mathcal{P}$. It is immediate to
show that $L^1_u(\mathcal{P})\subset L^1(\mathcal{P})$, and
$L^1(\{\mathbb{P}\})=L^1_u(\{\mathbb{P}\})=L^1(\mathbb{P})$. $L^1_u(\mathcal{P})$ is also a Banach
space endowed with the norm $\|\cdot\|_{1,\mathcal{P}}$, hence $L^1_u(\mathcal{P})$ is
closed in $L^1(\mathcal{P})$. Moreover, we can show also that
\begin{equation}
  \label{eq:Equivalence1}
  X\in L^1_u(\mathcal{P})\quad\Leftrightarrow\quad \{XdP/d\mathbb{P}\}_{P\in\mathcal{P}}\text{ is uniformly integrable.}
\end{equation}
In particular, the set $\mathcal{D}$ is equivalently written as:
\begin{equation}
  \label{eq:Dalternative}
  \mathcal{D}=\{X\in L^\infty:\, f(\cdot, X)^+\in L^1_u(\mathcal{P})\}.
\end{equation}

\begin{rem}
  \label{rem:DenisHuPeng}
  The spaces $L^1(\mathcal{P})$ and $L^1_u(\mathcal{P})$ correspond respectively to
  $\mathbb{L}^1$ and $\mathbb{L}^1_b$ in
  \cite[Section~2.2]{denis11:_funct_spaces_capac_relat_sublin_expec} which
  provide some basic properties including a counter part of
  (\ref{eq:Equivalence1}) in a slightly different setting. In
  \citep{denis11:_funct_spaces_capac_relat_sublin_expec} $\Omega$ is a complete
  metric space with $\mathcal{F}$ being its Borel $\sigma$-field but $\mathcal{P}$ is
  not necessarily dominated by a single probability $\mathbb{P}$, while we
  assume $\mathcal{P}$ is dominated but impose no topological restriction to
  $\Omega$ and $\mathcal{F}$. Finally, we used the subscript ``$u$'' to keep
  the \emph{uniform integrability} in mind, while the subscript
  ``$b$'' in \citep{denis11:_funct_spaces_capac_relat_sublin_expec} comes from
  the fact that $\mathbb{L}_b$ is the completion under $\|\cdot\|_{1,\mathcal{P}}$ of
  the space of \emph{bounded} functions.
\end{rem}

\begin{ex}[$\mathcal{D}=\emptyset$ but $\mathrm{dom}(\mathcal{I}_{f,\mathcal{P}})=L^\infty$]
\label{ex:1}
We first give an extreme example following \citep[Example
20]{denis11:_funct_spaces_capac_relat_sublin_expec}. We take
$\Omega=\mathbb{N}$ with $\mathcal{F}=2^\mathbb{N}$. In this case, \emph{every} probability
measure is absolutely continuous w.r.t. $\mathbb{P}$, given by
$\mathbb{P}(\{n\})=2^{-n}$. For each $n\in\mathbb{N}$, we define $P_n$ by
  \begin{equation}
    \label{eq:CountEx1}
    P_n(\{1\})=1-1/n,\, P_n(\{n\})=1/n,\, P_n(\{k\})=0\text{ if }k\not\in\{1,n\}.
  \end{equation}
  Then we set $\mathcal{P}=\overline{\mathrm{conv}}(P_n;\, n\in\mathbb{N})$. This
  $\mathcal{P}$ is weakly compact. To see this, it suffices to show the
  uniform integrability of $\{P_n\}_n$ since closed convex hull of
  uniformly integrable family is again uniformly integrable (see
  \cite[Th. II, 20]{dellacherie_meyer78}). Noting that
  $dP_n/d\mathbb{P}=2(1-1/n)1_{\{1\}}+2^n/n1_{\{n\}}$,
  $E[(dP_n/d\mathbb{P})1_{\{dP_n/d\mathbb{P}\geq N\}}]=(1/n)1_{\{2^n/n\geq N\}}$ for
  every $N\geq 2$, hence $\sup_nE[(dP_n/d\mathbb{P})1_{\{dP_n/d\mathbb{P}\geq
    N\}}]=1/n_N$, where $n_N:=\min(n:\, 2^n/n\geq N)\rightarrow
  \infty$ as $N\rightarrow\infty$.

  Consider a random variable $W$ defined by $W(n)=n$. Then since
  $E_{P_n}[W]=1\cdot (1-1/n)+n\cdot (1/n)=2-1/n$, we see
  $\sup_{P\in\mathcal{P}}E_P[W]=\sup_nE_{P_n}[W]=2$.  But,
  \begin{align*}
    E_{P_n}[W1_{\{W\geq N\}}]=
    \begin{cases}
        0&\text{ if }n<N,\\
        1&\text{ if }n\geq N.
      \end{cases}
    \end{align*}
  Hence $\sup_nE_{P_n}[W1_{\{W\geq N\}}]=1\not\rightarrow 0$.

  Let $f(\omega,x)=W(\omega)e^x$. For any constant random variable
  $a$, $f(\cdot, a)^+=e^aW$, hence $\sup_{P\in\mathcal{P}}E_P[f(\cdot,
  X)^+]=\sup_nE_{P_n}[f(\cdot, X)^+]=e^a\sup_nE_{P_n}[W]=2e^a$, but
  $\sup_{P\in\mathcal{P}}E_P[f(\cdot,X)^+1_{\{f(\cdot,X)\geq
    N\}}]=e^a\sup_nE_{P_n}[W1_{\{e^aW\geq N\}}]=e^a$, for all
  $N\in\mathbb{N}$. This shows that $a\in
  \mathrm{dom}(\mathcal{I}_{f,\mathcal{P}})\setminus\mathcal{D}$. Moreover, we have
  $\emptyset= \mathcal{D}\neq \mathrm{dom}(\mathcal{I}_{f,\mathcal{P}})=L^\infty$. Indeed,
  since $x\mapsto f(\omega,x)$ is increasing, we have for all $X\in
  L^\infty$,
  \begin{align*}
    & \sup_nE_{P_n}[f(\cdot, X)^+]\leq \sup_n E_{P_n}[f(\cdot,
    \|X\|_\infty)^+]= 2e^{\|X\|_\infty},\\
    &\sup_nE_{P_n}[f(\cdot,X)^+1_{\{f(\cdot,X)^+\geq
      N\}}]\\
    &\qquad\geq\sup_nE_{P_n}[f(\cdot,
    -\|X\|_\infty)^+1_{\{f(\cdot,-\|X\|_\infty)^+\geq
      N\}}]=e^{-\|X\|_\infty}.
  \end{align*}
  Therefore, $X\in\mathrm{dom}(\mathcal{I}_{f,\mathcal{P}})$ for all $X\in L^\infty$,
  but $X\not\in\mathcal{D}$ for all $X\in L^\infty$.
  
\end{ex}

\begin{ex}\label{ex:2}
  Let $W$ be same as the previous example, and set
  $f(\omega,x):=W(\omega)^x$. Since $W\geq 1$, this $f$ is
  well-defined. For $X\equiv 1$, we have $f(\cdot,X)\in
  L^1(\mathcal{P})\setminus L^1_u(\mathcal{P})$. On the other hand, if
  $\|X\|_\infty\leq \gamma<1$,
  \begin{align*}
    \sup_{P\in\mathcal{P}}E_P[f(\cdot, X)1_{\{|f(\cdot,X)|\geq n\}}]&\leq\sup_{P\in\mathcal{P}}E_P[W^\gamma 1_{\{W^\gamma\geq n\}}]\\
    &\leq \sup_{P\in\mathcal{P}}E_P[W]^\gamma
    \cdot\sup_{P'\in\mathcal{P}}P'(W^\gamma\geq n)^{1-\gamma} \stackrel{n\rightarrow \infty}{\rightarrow}0.
  \end{align*}
  The last convergence comes from the uniform integrability of
  $\mathcal{P}$. Therefore, we have $X\in \mathcal{D}$, and consequently,
  $\emptyset\neq\mathcal{D}\subsetneq \mathrm{dom}(\mathcal{I}_{f,\mathcal{P}})$.
\end{ex}

Although the inequalities in (\ref{eq:EstConj}) can generally be
strict, we can still have an exact estimate in the case corresponding
to the second assertion of Theorem~\ref{thm:RockafellarClassical}.

\begin{cor}
  \label{cor:RepIntFuncRobConti}
  If $\mathcal{D}= L^\infty$, $\mathcal{I}_{f,\mathcal{P}}$ is norm continuous on the whole
  $L^\infty$, and
    \begin{equation}
      \label{eq:ConjReg}
      \sup_{X\in L^\infty}(\nu(X)-\mathcal{I}_{f,\mathcal{P}}(X))=
      \begin{cases}
        \mathcal{J}_{\tilde f,\mathcal{P}}(d\nu/d\mathbb{P})&\quad\text{if $\nu$ is $\sigma$-additive,}\\
        +\infty&\quad\text{otherwise.}
      \end{cases}
    \end{equation}
 
\end{cor}
\begin{proof}
  The \emph{equality} (\ref{eq:ConjReg}) is again a direct consequence
  of (\ref{eq:EstConj}) since
  \begin{align*}
    \sup_{X\in L^\infty}\nu_s(X)=
    \begin{cases}
        0&\text{ if }\nu_s=0,\\
        +\infty&\text{ if }\nu_s\neq 0.
      \end{cases}
    \end{align*}
  The assumption $\mathcal{D}= L^\infty$ implies also that $\mathcal{I}_{f,\mathcal{P}}$ is
  finite on the whole $L^\infty$, thus the continuity follows from the
  general fact that a lower semicontinuous convex function on a Banach
  space is continuous on the interior of its effective domain, which
  is equal to $L^\infty$ in this case. See e.g.,
  \cite{ekeland_temam76}.
\end{proof}

An important question is when the condition $\mathcal{D}=L^\infty$ holds. A
trivial case is that $f$ is ``deterministic'', i.e., $f$ does not
depend on $\omega$. Indeed, for any $X\in L^\infty$, the random
variable $f(X)$ is again bounded, since we are assuming
$\mathrm{dom}(f)=\mathbb{R}$, hence $f$ is continuous. In the general case
with $f$ being ``random'', $x\mapsto f(\omega,x)$ is still continuous,
but the bound $\sup_{|x|\leq\|X\|_\infty}f(\omega,x)$ depends on
$\omega$. The next criterion is easy, but will turn out to be useful.
\begin{cor}
  \label{cor:RandPlusCont}
  Suppose that there exists a continuous function
  $g:\mathbb{R}\rightarrow\mathbb{R}$ and a random variable $W$ such that
  $\{WdP/d\mathbb{P}\}_{P\in\mathcal{P}}$ is uniformly integrable, and
  \begin{equation}
    \label{eq:CorSuffCondContiRand}
    f(\omega,x)\leq g(x)+W(\omega).
  \end{equation}
  Then $\mathcal{D}=L^\infty$ and hence~(\ref{eq:ConjReg}) holds.
\end{cor}
\begin{proof}
  It suffices to note that a family $\mathcal{X}$ in $L^1$ is uniformly
  integrable if it is dominated by a uniformly integrable family $\mathcal{Y}$
  in the sense that for every $X\in \mathcal{X}$, there exists $Y\in \mathcal{Y}$ such
  that $|X|\leq |Y|$. In the present case, we take
  $\mathcal{Y}=\{(|g(X)|+|W|)dP/d\mathbb{P}\}_{P\in\mathcal{P}}$ which is uniformly
  integrable by the assumptions, and dominates
  $\{f(\cdot,X)^+dP/d\mathbb{P}\}_{P\in\mathcal{P}}$.
\end{proof}

\subsection{Proof of Theorem~\ref{thm:PointWise}}
\label{sec:ProofIntRepMain}

\begin{lem}
  \label{lem:IntSelection}
  Let $D$ be a random variable with $D^-\in L^1$, and $\alpha$ be a
  constant with $\alpha<E[D]$. Then we can take some $Z\in L^1$ such
  that $Z< D$, a.s., and $E[Z]>\alpha$.
\end{lem}
\begin{proof}
  Take $\varepsilon>0$ with $E[D]-\varepsilon>\alpha$. Then for any
  $\gamma\in\mathbb{R}$, $(D-\varepsilon)\wedge \gamma<D$,
  $(D-\varepsilon)\wedge \gamma\in L^1$ since $D^-\in L^1$, and
  $E[(D-\varepsilon)\wedge \gamma]\nearrow E[D]-\varepsilon>\alpha$ by
  the monotone convergence theorem. Therefore,
  $Z:=(D-\varepsilon)\wedge \gamma_0$ satisfies the condition of the
  statement for a large $\gamma_0$.
\end{proof}

\begin{lem}[{\cite[Lemma 6]{rockafellar68:_integ_whichI}}]
  \label{lem:MeasSelect} Let $g$ be a normal convex integrand, and $Z$
  be a random variable such that
  \begin{equation}
    \label{eq:MeasSelect}
    \inf_{x\in \mathbb{R}}g(\cdot,x)<Z,\quad \text{a.s.}
  \end{equation}
  There then exists some finite valued random variable $X$ such that
  \begin{equation}
    \label{eq:MeasSelect2}
    g(\cdot,X)\leq Z,\quad \text{a.s.}
  \end{equation}

\end{lem}

\begin{lem}
  \label{lem:USC}
  For any $X\in \mathcal{D}$, the map $P\mapsto E_P[f(\cdot,X)]$ is weakly
  upper semicontinuous on $\mathcal{P}$.
\end{lem}
\begin{proof}
  Let $A_\alpha:=\{P\in\mathcal{P}:\,E_P[f(\cdot,X)]\geq \alpha\}$, which is
  convex, hence is weakly closed if and only if strongly closed. Thus,
  it suffices to show that $A_\alpha$ is strongly closed for all
  $\alpha\in\mathbb{R}$. Let $(P_n)$ be a convergent sequence in $A_\alpha$,
  i.e., $dP_n/d\mathbb{P}\rightarrow dP/d\mathbb{P}$ in $L^1$ for some probability
  measure $P\in\mathcal{P}$, and we show that $P\in A_\alpha$. Taking a
  subsequence if necessary, we may assume the a.s. convergence. For
  each $n$,
  \begin{align*}
    \frac{dP_n}{d\mathbb{P}}f(\cdot,X)\leq \frac{dP_n}{d\mathbb{P}}f(\cdot,X)^+,
  \end{align*}
  and the family $\{f(\cdot,X)^+dP_n/d\mathbb{P}\}_n$ is uniformly integrable
  and a.s. convergent. Therefore, we can apply (reverse) Fatou's lemma
  to get:
  \begin{align*}
    \alpha\leq \limsup_n E_{P_n}[f(\cdot,X)]\leq E_P[f(\cdot, X)].
  \end{align*}
  Hence $P\in A_\alpha$ and the proof is complete.
\end{proof}

\begin{proof}[Proof of Theorem~\ref{thm:PointWise}]
  We start from the second inequality. Noting that
  $\nu(X)=\nu_r(X)+\nu_s(X)=E[Xd\nu_r/d\mathbb{P}]+\nu_s(X)$, the inequality
  (\ref{eq:YoungIneqRob}) in Proposition~\ref{prop:SupIntWellDef}
  shows that
  \begin{align*}
    \sup_{X\in L^\infty}(\nu(X)-\mathcal{I}_{f,\mathcal{P}}(X))%
    &=\sup_{X\in\mathrm{dom}(\mathcal{I}_{f,\mathcal{P}})}(\nu(X)-\mathcal{I}_{f,\mathcal{P}}(X))\\
    &=\sup_{X\in\mathrm{dom}(\mathcal{I}_{f,\mathcal{P}})}\left(E[Xd\nu_r/d\mathbb{P}]-\mathcal{I}_{f,\mathcal{P}}(X)+\nu_s(X)\right)\\
    &\stackrel{\text{~(\ref{eq:YoungIneqRob})}}{\leq} \mathcal{J}_{\tilde
      f,\mathcal{P}}(d\nu_r/d\mathbb{P})+\sup_{X\in\mathrm{dom}(\mathcal{I}_{f,\mathcal{P}})}\nu_s(X).
  \end{align*}
  
  The first inequality is more subtle.  Observe that $X\mapsto
  \nu(X)-E_P[f(\cdot,X)]$ is concave on $\mathcal{D}$, and $P\mapsto
  \nu(X)-E_P[f(\cdot,X)]$ is convex and lower semicontinuous on the
  weakly compact set $\mathcal{P}$ for all $X\in \mathcal{D}$ by
  Lemma~\ref{lem:USC}. Thus a minimax theorem shows:
  \begin{align*}
    \sup_{X\in L^\infty}(\nu(X)-\mathcal{I}_{f,\mathcal{P}}(X))&=\sup_{X\in L^\infty}\inf_{P\in\mathcal{P}}(\nu(X)-E_P[f(\cdot, X)])\\
    &\geq \sup_{X\in\mathcal{D}}\inf_{P\in\mathcal{P}}(\nu(X)-E_P[f(\cdot,X)])\\
    &= \inf_{P\in\mathcal{P}}\sup_{X\in\mathcal{D}}(\nu(X)-E_P[f(\cdot,X)]).
  \end{align*}

  We shall show:

  \noindent \textbf{Claim.} For any $\alpha<\mathcal{J}_{\tilde
    f,\mathcal{P}}(d\nu_r/d\mathbb{P})$ and $\beta <\sup_{X\in \mathcal{D}}\nu_s(X)$, we have
  \begin{equation}
    \label{eq:ClaimProofMainConj1}
    \sup_{X\in\mathcal{D}}(\nu(X)-E_P[f(\cdot,X)])> \alpha +\beta,\quad\forall P\in\mathcal{P}.
  \end{equation}
  \noindent \emph{Proof of Claim.} Note first that there exists by
  definition an element $X_s\in\mathcal{D}$ with $\nu_s(X_s)>\beta$. Also,
  there exists an increasing sequence $(A_n)$ in $\mathcal{F}$ such that
  $\mathbb{P}(A_n)\nearrow 1$ and $|\nu_s|(A_n)=0$ for each $n$, by the
  singularity of $\nu_s$. In particular, for any $X\in L^\infty$,
  $\nu_s(X1_{A_n}+X_s1_{A_n^c})=\nu_s(X_s)>\beta$.

  As for the regular part, since $\alpha<\mathcal{J}_{\tilde
    f,\mathcal{P}}(d\nu_r/d\mathbb{P})=\inf_{P\in\mathcal{P}}E[\tilde f(\cdot,
  d\nu_r/d\mathbb{P})]$, Lemma~\ref{lem:IntSelection} shows the existence,
  for each $P\in\mathcal{P}$, of an integrable random variable $Z_P$ such that
  \begin{equation}
    \label{eq:ProofClaimConj2}
    \begin{split}
      E[Z_P]&>\alpha\text{ and }\\
      Z_P&<\tilde f\left(\cdot,
        \frac{d\nu_r}{d\mathbb{P}},\frac{dP}{d\mathbb{P}}\right)
      =\sup_{x\in\mathbb{R}}\left(x\frac{d\nu_r}{d\mathbb{P}}-\frac{dP}{d\mathbb{P}}f(\cdot,x)\right),\text{
        a.s.}
    \end{split}
  \end{equation}
  The latter condition and Lemma~\ref{lem:MeasSelect} applied to the
  normal integrand $(\omega,x)\mapsto
  f(\omega,x)(dP/d\mathbb{P})(\omega)-x(d\nu_r/d\mathbb{P})(\omega)$ yields a
  measurable selection $X^0_P\in L^0$ with
  \begin{equation}
    \label{eq:ProofClaimConj3}
    X_P^0d\nu_r/d\mathbb{P}-    f(\cdot,X^0_P)dP/d\mathbb{P}\geq Z_P.
  \end{equation}
  Note that $X^0_P$ is not an element of $\mathcal{D}$ (not even in
  $L^\infty$) in general. Thus we need to \emph{approximate} $X^0_P$
  with elements of $\mathcal{D}$. Recall that we are assuming $f$ to be finite
  valued. Let $B_n:=\{|X^0_P|\leq n\}\cap \{|f(\cdot,X^0_P)|\leq n\}$,
  and $C_n:=A_n\cap B_n$, for each $n$. Then $\mathbb{P}(C_n)\nearrow 1$,
  $|\nu_s|(C_n)=0$, and
  \begin{align*}
    X_P^n:=X_P^01_{C_n}+X_s1_{C_n^c}\in\mathcal{D},\text{ for each $n$.}
  \end{align*}
  Indeed, $X^0_P1_{C_n}$ and $1_{C_n}f(\cdot,X^0_P)$ are bounded by
  the construction, $f(\cdot, X_P^n)=1_{C_n}f(\cdot,
  X^0_P)+1_{C_n^c}f(\cdot,X_s)$, and hence the family
  $\{f(\cdot,X_P^n)^+dP/d\mathbb{P}\}_{P\in\mathcal{P}}$ is uniformly integrable by
  the weak compactness of $\mathcal{P}$ and the uniform integrability of
  $\{f(\cdot,X_s)^+dP/d\mathbb{P}\}_{P\in\mathcal{P}}$. We have
  \begin{align*}
    E[&X^n_Pd\nu_r/d\mathbb{P}]-E_P[f(\cdot,X_P^n)]\\
    &=E\left[1_{C_n}\left(X_P^0\frac{d\nu_r}{d\mathbb{P}}-\frac{dP}{d\mathbb{P}}f(\cdot,X^0_P)\right)\right]
    +E\left[1_{C_n^c}\left(X_s\frac{d\nu_r}{d\mathbb{P}}-\frac{dP}{d\mathbb{P}}f(\cdot,X_s)\right)\right]\\
    &\geq E[1_{C_n}Z_P]+E\left[1_{C_n^c}\left(X_s\frac{d\nu_r}{d\mathbb{P}}-\frac{dP}{d\mathbb{P}}f(\cdot,X_s)\right)\right]\\
    &=E[Z_P]+E[1_{C_n^c}\Xi_P],
  \end{align*}
  where $\Xi_P:=X_sd\nu_r/d\mathbb{P}-f(\cdot,X_s)dP/d\mathbb{P}-Z_P\in L^1$. Since
  $\nu_s(X_P^n)=\nu_s(X_s)>\beta$, 
  \begin{align*}
    \sup_{X\in\mathcal{D}}(\nu(X)-E_P[f(\cdot,X)])&\geq
    E[X_P^nd\nu_r/d\mathbb{P}]-E_P[f(\cdot,X^n_P)]+\nu_s(X_P^n)\\
    &\geq E[Z_P]+\nu_s(X_s)+E[1_{C_n^c}\Xi_P],
  \end{align*}
  for each $n$. Since $\lim_nE[1_{C_n^c}\Xi_P]=0$, we have
  \begin{align*}
    \sup_{X\in\mathcal{D}}(\nu(X)-E_P[f(\cdot,X)])\geq
    E[Z_P]+\nu_s(X_s)>\alpha+\beta,
  \end{align*}
  and the claim is proved.

  We now complete the proof of the first inequality in
  (\ref{eq:EstConj}). By taking $\alpha=\mathcal{J}_{\tilde
    f,\mathcal{P}}(d\nu_r/d\mathbb{P})-\varepsilon/2$ and
  $\beta=\sup_{X\in\mathcal{D}}\nu_s(X)-\varepsilon/2$, we have
  \begin{align*}
    \inf_{P\in\mathcal{P}}\sup_{X\in\mathcal{D}}&(\nu(X)-E_P[f(\cdot, X)])\\
    &\geq (\mathcal{J}_{\tilde f,\mathcal{P}}(d\nu_r/d\mathbb{P})-\varepsilon/2)
    +(\sup_{X\in\mathcal{D}}\nu_s(X)-\varepsilon/2)\\
    &=\mathcal{J}_{\tilde
      f,\mathcal{P}}(d\nu_r/d\mathbb{P})+\sup_{X\in\mathcal{D}}\nu_s(X)-\varepsilon,
  \end{align*}
  for all $\varepsilon>0$, and the proof is complete.
\end{proof}

\section{Proof of the Duality}
\label{sec:ProofMain}

We now complete our program outlined in Section~\ref{sec:Outline}.  We
start by translating the context of robust utility maximization into
the language of Section~\ref{sec:PointWise}. Set
\begin{equation}
  \label{eq:ProofDualityF}
  f(\omega,x):=-U(-x+B(\omega)), \quad \forall (\omega,x)\in\Omega\times\mathbb{R},
\end{equation}
which is $\mathcal{F}\otimes\mathcal{B}(\mathbb{R})$-measurable, and $x\mapsto f(\omega,x)$
is convex and continuous, hence is a finite valued normal convex
integrand in the sense of Definition~\ref{dfn:NormalIntegrands}. The
conjugate of $f$ is given by $f^*(\cdot, y)=V(y)+yB$, and
\begin{equation}
  \label{eq:ftilde}
  \tilde f(\omega,y,z)=
  \begin{cases}
    0&\text{ if }y=z=0,\\
    +\infty&\text{ if }y\neq 0, z=0,\\
    zV(y/z)+yB(\omega)&\text{ if }z>0.
  \end{cases}
\end{equation}
Noting that $u_{B,\mathcal{P}}(X)=-\mathcal{I}_{f,\mathcal{P}}(-X)$, and
$v_{B,\mathcal{P}}=(\mathcal{I}_{f,\mathcal{P}})^*$, we now arrive at the position to prove
the \emph{key lemma}.

\begin{proof}[Proof of Lemma~\ref{lem:Key}]
  We shall apply Corollary~\ref{cor:RandPlusCont} to $(f,\tilde f)$
  given by (\ref{eq:ProofDualityF}) and (\ref{eq:ftilde}).  Since
  $f^*(\cdot, 1)=V(1)+B$, Remark~\ref{rem:DirectAss}.1, and
  Remark~\ref{rem:IntegFuncAss}.2 guarantee that
  (\ref{eq:IntegRobust2}) is satisfied. On the other hand, the
  concavity of $U$ implies that
  \begin{align*}
    f(\cdot,x)\leq
    -\frac{\varepsilon}{1+\varepsilon}U\left(-\frac{1+\varepsilon}{\varepsilon}x\right)-\frac{1}{1+\varepsilon}U(-(1+\varepsilon)B^-)
    =:g(x)+W.
  \end{align*}
  Here $g$ is continuous and finite on $\mathbb{R}$, while
  $\{WdP/d\mathbb{P}\}_{P\in\mathcal{P}}$ is uniformly integrable by (A4). Hence we
  can apply Corollary~\ref{cor:RandPlusCont}, and $u_{B,\mathcal{P}}$ is
  continuous in particular.

  It remains to compute the explicit form of $\mathcal{J}_{\tilde f,\mathcal{P}}$ on
  $ba^\sigma_+$. We first show:
  \begin{equation}
    \label{eq:ExpFtilde}
    E\left[\tilde f\left(\cdot, \frac{d\nu}{d\mathbb{P}},\frac{dP}{d\mathbb{P}}\right)\right]
    =
    \begin{cases}
      V(\nu|P)+\nu(B)&\quad\text{if }V(\nu|P)<\infty\\
      +\infty&\quad\text{otherwise.}
    \end{cases}
  \end{equation}
  On the set $\{d\nu/d\mathbb{P}>0,dP/d\mathbb{P}=0\}$, we have $\tilde
  f(\cdot,d\nu/d\mathbb{P},dP/d\mathbb{P})=+\infty$, while on $\{dP/d\mathbb{P}>0\}$,
  \citep[][Lemma~3.4]{owari10:_APFM} shows
  \begin{equation}
    \label{eq:ElemEst1}
    \begin{split}
      & \frac{\varepsilon}{1+\varepsilon}\frac{dP}{d\mathbb{P}}\left(
        V\left(\frac{d\nu}{dP}\right)-V(1)\right)+\frac{dP}{d\mathbb{P}}U(-(1+\varepsilon)B^-)\\
      &\qquad \leq \tilde f \left(\frac{d\nu}{d\mathbb{P}},\frac{dP}{d\mathbb{P}}\right)%
      \leq\frac{1+\varepsilon}{\varepsilon}\frac{dP}{d\mathbb{P}}
      V\left(\frac{d\nu}{dP}\right)
      -\frac{1}{\varepsilon}\frac{dP}{d\mathbb{P}}U(-\varepsilon B^+).
    \end{split}
  \end{equation}
  Recalling that $V$ is bounded from below, the integrability
  assumption (A4) implies that $\tilde f(\cdot,
  d\nu/d\mathbb{P},dP/d\mathbb{P})^-\in L^1$ for all $\nu\in ba_+^\sigma$, and
  $\tilde f(\cdot, d\nu/d\mathbb{P},dP/d\mathbb{P})\in L^1$ if and only if
  $V(\nu|P)<\infty$ (which implies $\nu\ll P$). Further, when the
  latter condition is satisfied, we have $B\in L^1(\nu)$, hence
  $E[\tilde f(\cdot,
  d\nu/d\mathbb{P},dP/d\mathbb{P})]=E[(dP/d\mathbb{P})\{V(d\nu/dP)+(d\nu/d\mathbb{P})B\}]=V(\nu|P)+\nu(B)$. We
  thus obtain (\ref{eq:ExpFtilde}), and taking the infimum over $\mathcal{P}$,
  we have for all $\nu\in ba^\sigma_+$,
\begin{equation}
    \label{eq:RepFuncJ1}
    \mathcal{J}_{\tilde f,\mathcal{P}}(d\nu/d\mathbb{P})=
    \begin{cases}
      V(\nu|\mathcal{P})+\nu(B)&\quad\text{if }V(\nu|\mathcal{P})<\infty\\
      +\infty&\quad\text{otherwise.}
    \end{cases}
  \end{equation}
  This conclude the proof of lemma.
\end{proof}

\begin{cor}
  \label{cor:KeyLemma}
  For every $Q\in\mathcal{M}_V$, we have
  \begin{equation}
    \label{eq:KeyModified}
    v_{B,\mathcal{P}}(\lambda Q)=V(\lambda Q|\mathcal{P})+\lambda E_Q[B],\quad\forall \lambda\geq 0.
  \end{equation}

\end{cor}
\begin{proof}
  Let $Q\in \mathcal{M}_V$, then $B\in L^1(Q)$ by Remark~\ref{rem:DirectAss},
  and hence the right hand side is well-defined for all $\lambda\geq
  0$. Since $v_{B,\mathcal{P}}(\lambda Q)<\infty$ if and only if $V(\lambda
  Q|\mathcal{P})<\infty$, (\ref{eq:KeyModified}) is true in this case, and
  otherwise, both sides are $+\infty$.
\end{proof}

Recall that the convex cone $\mathcal{C}$ is defined by (\ref{eq:SetC}). We
need the following lemma, which guarantees that the element $0\in
ba_+^\sigma$ never contributes to the dual problem.

\begin{lem}
  \label{lem:Variational}
  We have
  \begin{equation}
    \label{eq:Variational1}
    \inf_{\lambda>0}\inf_{Q\in\mathcal{M}_V}\left(V(\lambda Q|\mathcal{P})+\lambda E_Q[B]\right)<V(0).
  \end{equation}
\end{lem}

\begin{proof}
  This is trivial if $V(0)=+\infty$, thus we assume $V(0)<\infty$.
  Taking a pair $(\bar Q,\bar P)\in\mathcal{M}_V\times\mathcal{P}$ as well as
  $\bar\lambda>0$ with $\bar Q\sim \bar P\sim \mathbb{P}$ and $V(\bar\lambda
  \bar Q|\bar P)<\infty$ by Remark~\ref{rem:NA1}, the result will
  follow if we can show $V(0)>E^{\bar P}[V(\lambda d\bar Q/d\bar
  P)+\lambda(d\bar Q/d\bar P)B]$ for some $\lambda>0$.  Set
  $\varphi(\lambda):=E^{\bar P}[V(\lambda d\bar Q/d\bar
  P)+\lambda(d\bar Q/d\bar P)B]$, which is convex and finite on
  $\lambda \in [0,\bar\lambda]$ since $V(0)<\infty$. It is easy (cf.
  the proof of \cite[Lemma 3.7]{owari10:_APFM}) to show that $\varphi$
  is differentiable, and $\varphi'(\lambda)=E^{\bar Q}[V'(\lambda
  d\bar Q/d\bar P)+B]$, hence $\varphi'(0)=V'(0)+E^{\bar
    Q}[B]=-\infty$ by (\ref{eq:InadaV})2. Thus, we have
  $\varphi(\lambda)<V(0)$ for some $\lambda>0$.
\end{proof}

The next one is an abstract version of the duality.

\begin{prop}
  \label{prop:AbstDuality}
  Under the assumptions of Theorem~\ref{thm:MainFinancial}, we have
  \begin{equation}
    \label{eq:DualityAbst}
    \sup_{X\in\mathcal{C}}u_{B,\mathcal{P}}(X)=\min_{\lambda> 0,Q\in\mathcal{M}_V}(V(\lambda Q|\mathcal{P})+\lambda E_Q[B])<\infty.
  \end{equation}

\end{prop}
\begin{proof}
  Since $u_{B,\mathcal{P}}$ is continuous on the whole $L^\infty$, Fenchel's
  duality theorem (\cite[Theorem 1]{rockafellar66:_fenchel}) shows that
  \begin{align}\label{eq:ProofDualAbstract}
    \sup_{X\in\mathcal{C}}u_{B,\mathcal{P}}(X)=\min_{\nu\in
      \mathcal{C}^\circ}v_{B,\mathcal{P}}(\nu)=\min_{\nu\in\mathcal{C}^\circ\cap\mathrm{dom}(v_{B,\mathcal{P}})}v_{B,\mathcal{P}}(\nu).
  \end{align}
  Here ``$\min$'' means of course that it is attained by some
  $\hat\nu\in\mathcal{C}^\circ\cap \mathrm{dom}(v_{B,\mathcal{P}})$. By
  Lemma~\ref{lem:Key}, $\nu\in \mathrm{dom}(v_{B,\mathcal{P}})$ if and only if
  $\nu$ is $\sigma$-additive and $V(\nu|\mathcal{P})<\infty$, while every
  $\sigma$-additive element $\nu\in\mathcal{C}^\circ$ is a positive multiple
  of some local martingale measure, i.e., $\nu=\lambda Q$ with
  $\lambda\geq 0$ and $Q\in\mathcal{M}_{loc}$. Thus we can rewrite the right
  hand side of (\ref{eq:ProofDualAbstract}) using
  (\ref{eq:KeyModified}):
  \begin{align*}
    \min_{\nu\in\mathcal{C}^\circ\cap\mathrm{dom}(v_{B,\mathcal{P}})}v_{B,\mathcal{P}}(\nu) &=\min_{\lambda\geq 0,Q\in \mathcal{M}_V}v_{B,\mathcal{P}}(\lambda Q)\\
    &=\min_{\lambda \geq 0,Q\in\mathcal{M}_V}(V(\lambda Q|\mathcal{P})+\lambda E_Q[B]).
  \end{align*}
  But Lemma~\ref{lem:Variational} implies that $\lambda=0$ never
  contributes to the minimum, hence we obtain~(\ref{eq:DualityAbst}),
  and the finiteness of the right hand side follows from (A3).
\end{proof}

The next step is to replace the ``$\inf_{X\in\mathcal{C}}$'' by the infimum
over stochastic integrals $\theta\cdot S_T$. By the definition of
$\mathcal{C}$, the inclusion $\Theta_{bb}\subset \Theta_V$ (see
(\ref{eq:admissible}) and (\ref{eq:ThetaV}) for definitions) as well
as the monotonicity of utility function, we have
\begin{equation}
  \label{eq:Monotone}
  \begin{split}
    \sup_{X\in\mathcal{C}}u_{B,\mathcal{P}}(X)&\leq
    \sup_{\theta\in\Theta_{bb}}\inf_{P\in\mathcal{P}}E_P[U(\theta\cdot S_T+B)]\\
    & \leq \sup_{\theta\in\Theta_V}\inf_{P\in\mathcal{P}}E_P[U(\theta\cdot
    S_T+B)].
  \end{split}
\end{equation}
Therefore, it suffices to show:
\begin{lem}
  \label{lem:DualIneq1}
  We have
  \begin{equation}
    \label{eq:ThetaLeqDual}
    \sup_{\theta\in\Theta_V}\inf_{P\in\mathcal{P}}E_P[U(\theta\cdot S_T+B)]
    \leq \inf_{\lambda>0}\inf_{Q\in\mathcal{M}_V}(V(\lambda Q|\mathcal{P})+\lambda E_Q[B]).
  \end{equation}
\end{lem}

\begin{proof}
  Let $\theta\in\Theta_V$ be fixed and take any pair
  $(Q,P)\in\mathcal{M}_V\times\mathcal{P}$ as well as $\lambda>0$ with $V(\lambda
  Q|P)<\infty$ ($\Rightarrow$ $Q\ll P$). Then by Young's inequality,
  \begin{align*}
    U(\theta\cdot S_T+B)&\leq V\left(\lambda \frac{dQ}{dP}\right)+\lambda\frac{dQ}{dP}(\theta\cdot S_T+B)
  \end{align*}
  Since $\theta\cdot S$ is a $Q$-supermartingale (hence $\theta\cdot
  S_T$ is $Q$-integrable in particular), and $B\in L^1(Q)$, we obtain
  by taking $P$-expectation
  \begin{align*}
    E_P[U(\theta\cdot S_T+B)]%
    &\leq E_P\left[
      V\left(\lambda\frac{dQ}{dP}\right)+\lambda\frac{dQ}{dP}B\right]%
    +E_Q[\theta\cdot S_T]\\
    &\leq V(\lambda Q|P)+\lambda E_Q[B].
  \end{align*}
  When $(Q,P)\in\mathcal{M}_V\times\mathcal{P}$ but $V(\lambda Q|P)=\infty$, the right
  hand side is $+\infty$, hence this inequality is valid for all
  $\theta\in\Theta_V$, $\lambda >0$ and
  $(Q,P)\in\mathcal{M}_V\times\mathcal{P}$. Taking the infimum over $P\in\mathcal{P}$, we have
  $\inf_{P\in\mathcal{P}}E_P[U(\theta\cdot S_T+B)]\leq V(\lambda
  Q|\mathcal{P})+\lambda E_Q[B]$. Now taking the supremum over
  $\theta\in\Theta_V$ and infimum over
  $(\lambda,Q)\in(0,\infty)\times\mathcal{M}_V$, we obtain
  (\ref{eq:ThetaLeqDual}).
\end{proof}

\begin{proof}[Proof of Theorem~\ref{thm:MainFinancial}]
  
  The equality (\ref{eq:ThmDuality}) as well as the existence of a
  dual minimizer follow from (\ref{eq:DualityAbst}) and the inequality
  (\ref{eq:ThetaLeqDual}). Let $(\hat\lambda,\widehat{Q})\in
  (0,\infty)\times\mathcal{M}_V$ be a dual minimizer. Then since $\mathcal{P}$ is
  weakly compact, and $(\nu,P)\mapsto V(\nu|P)$ (hence $P\mapsto
  V(\hat\lambda \widehat{Q}|P)$) is weakly lower semicontinuous (\citep[Lemma
  2.7]{follmer_gundel06}), there exists a $\widehat{P}$ with $V(\hat\lambda
  \widehat{Q}|\mathcal{P})=V(\hat\lambda \widehat{Q}|\widehat{P})$. This triplet is a desired solution
  to the dual problem, i.e.,
  \begin{align*}
    V(\hat\lambda \widehat{Q}|\widehat{P})+\hat\lambda E_{\widehat{Q}}[B]=
    \inf_{\lambda>0}\inf_{(Q,P)\in\mathcal{M}_V\times\mathcal{P}}(V(\lambda
    Q|P)+\lambda E_Q[B]),
  \end{align*}
  and we finish the proof.
\end{proof}

\begin{proof}[Proof of Corollary~\ref{cor:Equiv}]
  Again we take a triplet $(\bar\lambda, \bar Q,\bar
  P)\in(0,\infty)\times\mathcal{M}_V\times\mathcal{P}$ satisfying
  $V(\bar\lambda \bar Q|\bar P)<\infty$ and $\bar Q\sim\bar
  P\sim\mathbb{P}$ by Remark~\ref{rem:NA1}, then set
  $\nu_\alpha:=\alpha\bar\lambda \bar Q+(1-\alpha)\hat\lambda \widehat{Q}$ and
  $P_\alpha :=\alpha \bar P+(1-\alpha)\widehat{P}$. Note that $\nu_\alpha\sim
  P_\alpha\sim\mathbb{P}$ for all $\alpha\in (0,1]$, and the function
  $\alpha \mapsto V(\nu_\alpha|P_\alpha)+\nu_\alpha(B)$ is convex and
  finite valued on $[0,1]$. Therefore, this function is upper
  semicontinuous on $[0,1]$, which implies $ \inf_{\alpha\in
    (0,1]}V(\nu_\alpha|P_\alpha)+\nu_\alpha(B)\leq
  \limsup_{\alpha\searrow 0} V(\nu_\alpha|P_\alpha)+\nu_\alpha(B)\leq
  V(\hat\lambda\widehat{Q}|\widehat{P})+\hat\lambda E_{\widehat{Q}}[B]$. This implies that
  \begin{align*}
    \inf_{\lambda>0}\inf_{Q\in\mathcal{M}_V^e} (V(\lambda Q|\mathcal{P})+\lambda
    E_Q[B]) &\leq \inf_{\alpha\in (0,1]}(V(\nu_\alpha|P_\alpha)+\nu_\alpha(B))\\
    &\leq V(\hat\lambda \widehat{Q}|\widehat{P})+\hat\lambda E_{\widehat{Q}}[B]\\
    &=\inf_{\lambda>0}\inf_{Q\in\mathcal{M}_V}(V(\lambda Q|\mathcal{P})+\lambda
    E_Q[B]).
  \end{align*}
  This concludes the proof.
\end{proof}


\end{document}